\documentclass[acmsmall]{acmart}

\acmJournal{PACMPL}
\acmVolume{1}
\acmNumber{POPL} 
\acmArticle{1}
\acmYear{2020}
\acmMonth{1}
\acmDOI{} 
\startPage{1}

\setcopyright{none}

\usepackage{booktabs}   
\usepackage{subcaption} 

\usepackage{algorithm}
\usepackage[noend]{algorithmic}
\usepackage{amsmath}
\usepackage{amsfonts}
\usepackage{amssymb}
\usepackage{comment}
\usepackage{enumerate}
\usepackage{graphicx}
\usepackage{listings}
\usepackage{multirow}
\usepackage{url}
\usepackage{wrapfig}
\usepackage{subcaption}

\usepackage[graphicx]{realboxes}
\usepackage{tikz}
\usetikzlibrary{automata,positioning}
\usepackage{tikz-qtree}
\usepackage{float}
\usepackage{floatflt}

\usepackage{xspace}

\usepackage{appendix}

\newcommand{\STOP}{\mathsf{STOP}}
\newcommand*{\State}{S}
\newcommand*{\Start}{\iota}
\newcommand*{\Public}{\Sigma}
\newcommand*{\Private}{\Gamma}

\newcommand{\trans}[1]{\stackrel{#1}{\longrightarrow}}

\newcommand*{\Always}{\square}
\newcommand*{\Eventually}{\lozenge}
\newcommand*{\False}{\mathsf{False}}

\newcommand*{\Next}{\mathsf{ X }}
\newcommand*{\True}{\mathsf{True}}
\newcommand*{\Untill}{\mathsf{\,U\,}}

\newcommand{\I}{\mathcal{I}}
\renewcommand{\O}{\mathcal{O}}

\newcommand{\Lang}{\mathcal{L}}

\newcommand{\PSPACE}{\textsf{PSPACE}\xspace}
\newcommand{\NLOGSPACE}{\textsf{NLOGSPACE}\xspace}
\newcommand{\TWOEXPTIME}{\textsf{2-EXPTIME}\xspace}

\newcommand{\lNot}{\,\neg}
\newcommand{\lAnd}{\,\wedge\,}
\newcommand{\lOr}{\,\vee\,}
\newcommand{\lEquiv}{\,\equiv\,}

\newcommand{\acc}{\mathsf{accept}}
\newcommand{\rej}{\mathsf{reject}}

\newcommand{\Union}{\,\cup\,}
\newcommand{\Inter}{\,\cap\,}
\newcommand{\Diff}{\,\backslash\,}

\newcommand{\then}{\rightarrow}
\newcommand{\choice}{\;|\;}

\newcommand{\trace}{\mathsf{trace}}
\newcommand{\cL}{\mathcal{L}}

\newcommand{\fulltree}{\mathsf{fulltree}}
\newcommand{\proc}{\mathsf{proc}}

\newcommand{\Lend}{\mathsf{endL}}

\newcommand{\Sink}{\mathsf{Sink}}
\newcommand{\Fail}{\mathsf{Fail}}
\newcommand{\Efail}{\mathsf{Efail}}
\newcommand{\Esink}{\mathsf{Esink}}
\newcommand{\Eprivate}{\mathsf{Eprivate}}
\newcommand{\Enabled}{\mathsf{enabled}}
\newcommand{\geneprivate}{\mathsf{GenEprivate}}
\newcommand{\Joint}{\mathsf{Joint}}
\newcommand{\normtrans}{\mathsf{normalTrans}}
\newcommand{\noSynch}{\mathsf{noSynch}}

\newtheorem*{theorem*}{Theorem}

\newcommand{\elide}[1]{}

\begin{document}

\title[Synthesis of Coordination Programs]{Synthesis of Coordination Programs from Linear Temporal Specifications}         


\author{Suguman Bansal}
\affiliation{
  \institution{Rice University}           
  \city{Houston}
  \state{TX}
  \postcode{77005-1892}
  \country{USA}                    
}
\email{suguman@rice.edu}          

\author{Kedar S. Namjoshi}
\affiliation{
  \institution{Bell Labs, Nokia}            
  \city{Murray Hill}
  \state{NJ}
  \postcode{07974}
  \country{USA}                    
}
\email{kedar.namjoshi@nokia-bell-labs.com}          

\author{Yaniv Sa'ar}
\affiliation{
  \institution{Bell Labs, Nokia}            
  \city{Kfar Saba}
  \postcode{4464321}
  \country{Israel}                    
}
\email{yaniv.saar@nokia-bell-labs.com}          

\begin{abstract}
This paper presents a method for synthesizing a reactive program which coordinates the actions of a group of other reactive programs, so that the combined system satisfies a temporal specification of its desired long-term behavior. Traditionally, reactive synthesis has been applied to the construction of a stateful hardware circuit. This work is motivated by applications to other domains, such as the IoT (the Internet of Things) and robotics, where it is necessary to coordinate the actions of multiple sensors, devices, and robots. The mathematical model represents such entities as individual processes in Hoare's CSP model. Given a network of interacting entities, called an \emph{environment}, and a temporal specification of long-term behavior, the synthesis method constructs a \emph{coordinator} process (if one exists) that guides the actions of the environment entities so that the combined system is deadlock-free and satisfies the given specification. The main technical challenge is that a coordinator may have only \emph{partial knowledge} of the environment state, due to non-determinism within the environment, and environment actions that are hidden from the coordinator. This is the first method to handle both sources of partial knowledge, and to do so for arbitrary linear temporal logic specifications. It is shown that the coordination synthesis problem is \PSPACE-hard in the size of the  environment. A prototype implementation is able to synthesize compact solutions for a number of coordination problems.


\end{abstract}

\ccsdesc[500]{Theory of computation~Concurrency}
\ccsdesc[500]{Theory of computation~Formal languages and automata theory}
\ccsdesc[500]{Software and its engineering~Correctness}

\keywords{coordination, synthesis, temporal logic}  

\maketitle

\section{Introduction}
\label{Sec:Intro}
Coordination problems arise naturally in many settings. In a so-called ``smart'' building, various sensors, heating and cooling devices must work in concert to maintain comfortable conditions. In a fully automated factory, a number of robots with specialized capabilities must collaborate to carry out manufacturing tasks. Typically, the individual agents are {reactive} and a {centralized coordinator} provides the necessary guidance to carry out a task.

A coordination program must work in the presence of several complicating factors such as concurrency, asynchrony, and distribution; it should recover gracefully from agent failures and handle noisy sensor data. All this complicates the design of coordination programs. It is often the case, however, that the task itself can be specified easily and compactly. We therefore consider whether it is possible to \emph{automatically synthesize} a reactive coordination program from a description of the agents and a specification of the desired long-term system behavior. 

\begin{figure}[t]
	\begin{center}
		\includegraphics[scale=0.45]{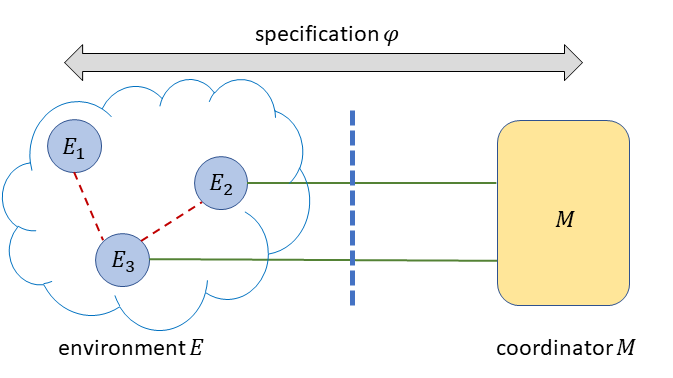}
	\end{center}
	\caption{The coordination model. The coordinator $M$ communicates with the agents $\{E_i\}$ through an interface (shown as solid green lines); it need not have a direct link to every agent. Agent-to-agent interactions (shown as dashed red lines) are hidden from the coordinator, as are actions internal to an agent (not shown). The specification $\varphi$, however, is ``all-seeing'' - it may refer to any action.}
	\label{fig:system}
\end{figure}

We formulate the coordination problem within the extensively studied framework of Communicating Sequential Processes (CSP)~\cite{DBLP:journals/cacm/Hoare78,DBLP:books/ph/Hoare85}, as illustrated in Fig.~\ref{fig:system}. Agents and the coordinator are represented as CSP processes.  Each process has a set of {\em private actions}, which are special to the process, and a set of {\em public actions}, which may be shared with other processes. The {\em coordination synthesis problem} considered here is defined as follows: given CSP processes $E_1, \ldots, E_n$, each modeling a reactive agent, and a temporal specification $\varphi$ over their public and private actions, construct a coordinator CSP process $M$ over the interface actions such that all computations of the combined system 
satisfy the specification $\varphi$. We represent the agents as a single \emph{environment} process $E$, formed by the parallel composition of the individual processes $E_1, \ldots, E_n$. An instance of the problem is {\em realizable} if there is such a CSP process $M$; otherwise, the given problem instance is {\em unrealizable}. The temporal specification is expressed in linear temporal logic (LTL)~\cite{pnueli1977temporal} or, more generally, as an automaton.

This formulation differs from the existing literature on reactive synthesis in several important ways. A major difference is the choice of CSP as a modeling language. Nearly all of the prior work is based on a \emph{synchronous, shared-variable} model of computation (cf.~\cite{church57,church62,buechi-landweber69,rabin69,pnueli1989synchsynthesis,kupferman2005safraless,bohy2012acacia+,gr1journal2012, schewe2013bounded-journal,alur2016compositional,hadaz2018reactive}). That is an appropriate model for hardware design but not for the coordination scenarios described above, simply because those systems are naturally asynchronous: there is no common clock. Pnueli and Rosner~\cite{pnueli1989synthesis} formulate an \emph{asynchronous} shared-variable model, but that is based on a highly adversarial scheduler and is quite weak as a result: for instance, the requirement that an input data stream is copied faithfully to the output has no solution in that model, while it has a simple solution in CSP.

A second difference is the communication model. In the prior work, communication is modeled by reads and writes to a shared memory. This is not a good fit for the targeted application domains, where the natural mode of communication is message passing. For instance, a coordinator may send a message to a robot asking it to lift up its arm, or to move to a given location. Implementing such handshakes in shared memory requires a dedicated synchronization protocol, complicating modeling as well as synthesis. On the other hand, such communication is easily and directly expressed in CSP, at a high level of atomicity that permits more specifications to be realized, as illustrated by the data-copying example. 

The new formulation crucially differs from prior work on synthesis for CSP~\cite{manna1981synthesis,wolper1982specification,d2013synthesizing,DBLP:journals/tse/CiolekBDPU17} in its treatment of hidden actions. The prior methods require that \emph{every} environment action is visible to a coordinator.  That can result in unnecessarily complex coordinators, which are forced to respond to every action, even if not all of those actions are relevant to achieving a goal. Moreover, full visibility goes against the key principle of information hiding in software design: indeed, the CSP language includes operators to limit the scope and visibility of actions, making models with hidden actions the common case.  

A coordinator must, therefore, operate correctly with only \emph{partial information} about the state of the environment agents. First, the state of an agent may be changed autonomously by a private action. Such transitions are hidden from the coordinator, introducing one form of partial information. Second, agents may interact and exchange information via a shared public action; such interactions are also unobserved by the coordinator, introducing yet another form of partial information. Finally, any interaction -- including one between an agent and the coordinator -- may have a non-deterministic effect, so that the precise next state is not known to the coordinator, introducing a third form of partial information.

The specification, on the other hand, is ``all-seeing'' and can refer to any action of the combined system. For instance, a specification may include a fairness assumption that rules out an infinite sequence of agent-to-agent interactions. 

Models with hidden actions alter the synthesis problem in fundamental ways, requiring the development of new algorithms. In the absence of hidden actions, the environment and the coordinator are synchronized. If the coordinator takes $n$ steps, the environment must also have taken $n$ steps. The inclusion of hidden actions introduces \emph{asynchrony}: if the coordinator has taken $n$ steps, the environment may have taken $n+m$ steps, where $m$ is the number of hidden  actions; $m$ is unknown to the coordinator and could be unbounded. Although the $m$ hidden actions are invisible to the coordinator, they could be referenced in the specification and cannot simply be ignored. The heart of our algorithm is a transformation that incorporates the effects of the hidden actions into the original temporal specification, producing a new specification that is expressed purely in terms of the interface actions. The transformed specification is then synthesized using existing synthesis methods. 

Our work handles specifications written in LTL, and carries out the specification transformation using automata-theoretic methods. The transformed automaton has a number of states that is \emph{linear} in the size of the automaton for the negated LTL specification and in the size of the environment model. We show how to express the transformation fully symbolically, as the number of transitions is exponential in the number of interface actions. This is in contrast to prior algorithms for CSP synthesis~\cite{DBLP:journals/tse/CiolekBDPU17}, which require explicit determinization steps that may introduce \emph{exponential} blowup in the number of states. The final synthesis step using the transformed specification is also carried out symbolically, via a translation to a Boolean SAT/QBF problem.

An LTL synthesis method is needed for this final step. The GR(1) subset of LTL is often used for specification, as it has efficient symbolic synthesis algorithms in the synchronous model~\cite{gr1journal2012,piterman2006synthesis}. The efficiency advantage is unfortunately lost under asynchrony, as it is not known whether GR(1) is closed under the specification transformation. Thus, even if the original specification is in GR(1), the transformed specification may lie outside GR(1), necessitating the use of a general LTL synthesizer.



Coordination synthesis has a high inherent complexity: we show that the question is
\PSPACE-hard in the size of $E$ for a fixed specification; in comparison,
model-checking $E$ is in \NLOGSPACE. This rather severe hardness result is not
entirely unexpected since problems in reactive synthesis typically exhibit high
worst-case complexity. For instance, synthesis in the Pnueli-Rosner model 
is \TWOEXPTIME-complete in the size of the LTL specification~\cite{pnueli1989synchsynthesis,pnueli1989synthesis}.
		
A prototype implementation of our algorithm  successfully synthesizes solutions to
several coordination problems. It uses Ordered Binary Decision
Diagrams (OBDDs)~\cite{bryant1986graph} for symbolic manipulation. The transformed specification is then checked for realizability
using the state-of-the-art symbolic synthesis tool BoSy~\cite{FFRT17}. We present case studies that synthesize a smart thermostat
and an arbiter for a number of concurrent processes. These studies illustrate the capabilities of the model and show
how synthesis improves the experience of designing coordination programs.
The prototype has limited scalability, the primary bottleneck being the capacity of the back-end synthesis tools.
Improvements to these tools (an active research topic) will therefore have a positive effect on the scalability of coordination synthesis.
The most impressive impact on scalability is, however, likely to come from entirely new methods for synthesis
that use symmetry, modularity, and abstraction strategies to effectively handle large state spaces,
drawing inspiration from the success of such techniques in formal verification.

\elide{

The prior work handles non-determinism but not hidden actions: as our model includes both, new synthesis methods are needed. 

For example, in the Thermostat (7.2.2, Fig 6-8), the coordinator is not exposed to actions “ACisOn” and “HeatisOn”. If hidden actions were not allowed (as in [D'Ippolito-et-al-2013]), the coordinator would have to respond to those actions, complicating its structure. Even in this small example, allowing internal actions contributes to simpler coordinator design. 

To support information hiding, CSP includes operators that limit the scope and visibility of actions, making models with internal actions the common case, not an exception. Our work is the first to show how to synthesize with such models.  

For instance, consider the following simple scenario. We have available an illumination sensor, a clock, a number of decorative outdoor lamps, and a timer, and would like to light up the lamps for, say, two hours after dusk. The task is easily described in temporal logic. For this coordination problem, it suffices to model each device as a small, abstract finite state machine. The illumination sensor model may have only two states, low and high; the clock model needs to be accurate only to the hour, and the timer model can have a similar coarse granularity; while each lamp is modeled as being either on or off. A reactive coordination program for this task would check the illumination sensor and the clock to determine when it is dusk, switch on the lamps and set the timer for two hours, switch off the lamps once the timer expires, and repeat the cycle the following day.

Although this work coins the term ``\emph{coordination synthesis}'' there is prior work on synthesizing reactive coordination programs. The core contributions of this work are an expressive framework for specifying coordination problems, and a new synthesis algorithm that constructs coordinators which operate with partial information of the environment state, while guiding the combined system so that it satisfies a specification from a rich class of temporal properties.

Recent work~\cite{d2013synthesizing} shows how to synthesize coordinators that can handle a particular form of partial information: non-determinism at an agent interface. That algorithm, however, is limited to specifications in a subclass of linear temporal logic called GR(1), while the model does not allow agents to have internal actions or private agent-to-agent interactions, which are a natural feature of many scenarios. Our new framework allows all these forms of interaction, which are hidden from the coordinator,  while the synthesis algorithm applies to specifications expressed as arbitrary LTL formulas or automata. In both respects, this represents a significant expansion of the class of solvable coordination problems.

The {\em second contribution} of this work is an algorithm to solve this problem.
To the best of our information, this is the first algorithm that takes all three sources of partial information into
account, and does so for arbitrary LTL specifications. (An in-depth
discussion of related work appears later in this paper.) We reduce coordination
synthesis to the ordinary synthesis of a new temporal specification which can be
represented as an automaton. This automaton, called the {\em specification
automaton}, combines the behavior of the reactive agents $E_1,\ldots,E_n$, 
in an asynchronous setting, with an automaton $\mathcal{A}$ for the LTL specification $\varphi$.

The key challenge in this construction is to represent the partial information of
the environment state. A naive encoding would use a subset
construction, producing an exponential blowup in the number of states of $E$.
We provide an alternate construction, where the number of states in our specification automaton grows \emph{linearly} with the number of states of $E$
and the number of states of the property automaton $\mathcal{A}$. 
Although the number of states grows linearly, the transition relation of the 
specification automaton has size exponential in the number of public
actions (i.e., interface) of the environment $E$. 
To ameliorate the exponential blowup incurred by the transition relation we use
symbolic constructions.
		
Towards a full understanding of the inherent difficulty of coordination
synthesis,  
our {\em third contribution} is to
prove a complexity-theoretic result showing that the question is
\PSPACE-hard in the size of $E$ for a fixed specification; in comparison,
model-checking is in \NLOGSPACE. This rather severe hardness result is not
entirely unexpected since problems in reactive synthesis typically exhibit high
worst-case complexity. For instance, both synchronous and asynchronous synthesis
for 
arbitrary LTL specifications are known to be \TWOEXPTIME-complete in the size
of the specification~\cite{pnueli1989synchsynthesis,pnueli1989synthesis}.
		
We also present a prototype implementation of our algorithm, which uses a
symbolic construction of the specification automaton, using Ordered Binary Decision
Diagrams (BDDs)~\cite{bryant1986graph}, and relies on the state-of-the-art tool
BoSy~\cite{FFRT17} for the final ordinary synthesis step. We present case
studies on synthesizing a smart thermostat, and arbiters for concurrent
processes. Our case studies illustrate the capabilities of the model, and show
how synthesis improves the experience of designing coordination
programs. While this prototype has limited scalability, additional
empirical analysis reveals that the primary bottleneck is in the current capabilities of the back-end synthesis tools.
Improvements to these tools (an active research topic) will therefore have a positive effect on the scalability of coordination synthesis.
The most impressive impact on scalability is, however, likely to come from entirely new synthesis methods which exploit symmetry, modularity, and abstraction,
drawing inspiration from the success of such techniques in formal verification. 

The classical models for synchronous and asynchronous synthesis, mentioned above, are both subsumed by this model (see Appendices ~\textsection~\ref{Sec:async-sim}-~\ref{Sec:sync-sim}).
}

\section{Illustrative Examples}
\label{Sec:MotivatingEg}

We present a series of small examples to illustrate the important features of the model and the considerations that must go into the design of a synthesis algorithm. This section focuses on giving a reader an intuitive view of the issues, precise definitions follow in later sections. 

\paragraph{CSP notation and semantics.}
The full CSP language has a rich structure (cf.~\cite{DBLP:books/ph/Hoare85,Roscoe-TPC-1997}); here we use only the most basic ``flat'' form of a CSP process, specified by a set of equations of the form below. 
\begin{equation*}
  P = a_0 \then Q_0 \choice a_1 \then Q_1 \choice \ldots \choice a_{n-1} \then Q_{n-1}
\end{equation*}
The meaning is that process $P$ evolves to process $Q_0$ on action $a_0$; to $Q_1$ on action $a_1$; and so forth. The actions need not be distinct: $P = a \then Q_0 \choice a \then Q_1$ represents a non-deterministic choice between $Q_0$ and $Q_1$ on action $a$. One may also view $P$ as a state in a state machine, with the equation specifying the transitions at state $P$: the machine moves to state $Q_0$ on action $a_0$; to state $Q_1$ on action $a_1$; and so forth. The special process $\STOP$ has no outgoing transitions at all and thus represents a dead-end state. An entire state machine is thus described by a set of interrelated equations, one for each state.

In CSP,  processes communicate only through an instantaneous, synchronized interaction on a common action. If processes $P$ and $Q$ have a common action $a$, and if $P$ may evolve to $P'$ on $a$ while $Q$ may evolve to $Q'$ on $a$, then their concurrent composition, denoted $P\parallel Q$, may \emph{jointly} evolve to $P'\parallel Q'$ through synchronization on action $a$. On the other hand, a process may also have internal actions, which are unsynchronized; so that if $P$ may evolve to $P''$ on an internal action $b$, then the composition $P\parallel Q$ may evolve to $P''\parallel Q$ on $b$.

\paragraph{Coordination Problems.} 
In the examples below, $E$ denotes the single environment process; we denote its public (i.e., interface) actions by $a_0,a_1,\ldots$, and its private (i.e., internal) action as $b$. The temporal specification of desired behavior, $\varphi$, is fixed to be ``finitely many $b$ actions''; this may be represented in temporal logic as $\Eventually\Always(\lNot b)$. The goal is to construct a coordinator process, $M$, such that the combined system, $E\parallel M$, is (1) free of deadlock, and (2) satisfies $\varphi$ on each of its infinite executions. If such a process exists, it is called a \emph{solution} to the coordination problem $(E,\varphi)$. A coordination problem is called \emph{realizable} if it has a solution, and \emph{unrealizable} otherwise. The examples illustrate both outcomes. 

\paragraph{Example 0: Non-blocking.}
The first example illustrates how a coordinator can guide the system so that it is deadlock-free. 
Consider $E$ defined as follows.
\begin{align*}
  & E    =  a_0 \then E_0  \choice a_1 \then \STOP \\
  & E_0  =  a_0 \then E_0
\end{align*}

Let $M = a_0 \then M$. As $E$ and $M$ may synchronize only on $a_0$, the composition $E \parallel M$ has only the single infinite computation $E\parallel M \trans{a_0} E_0 \parallel M \trans{a_0} E_0 \parallel M \ldots$, which trivially satisfies the specification. Hence $M$ is a solution to this coordination problem. 

 \begin{figure}[t]
 	\centering
 	\begin{minipage}{0.35\textwidth}
 		\centering
 		\begin{tikzpicture}[shorten >=1pt,node distance=2cm,on grid,auto] 
 		
 		\node[state,initial] (q)   { $E$}; 
 		\node[state] (q_0) [above right=of q] {$E_0$}; 
 		\node[state] (q_1) [right=of q] {$E_1$}; 
 		
 		\path[->] 
 		(q)   
 		edge node {\footnotesize $a_0$} (q_0)
 		edge node {\footnotesize $a_1$} (q_1)

 		(q_0) edge  [loop right] node { $a_0$} ()

 		(q_1) edge  [loop right] node { $b$} ();
 		
 		\end{tikzpicture}
 		\caption{Example 1. Public actions $a_0$, $a_1$; Private action $b$}
 		\label{Fig:Example1}
 	\end{minipage}
 	\hfill
 	\centering
 	\begin{minipage}{0.25\textwidth}
 		\centering
 		\begin{minipage}{\textwidth}
 			\centering
 			\begin{tikzpicture}[shorten >=1pt,node distance=1cm,on grid,auto] 
 			
 			\node[state,initial] (q)   {  $M$};

 			\path[->] 
 			(q)   
 			(q) edge  [loop right] node {  $a_0$} ();
 	
 			\end{tikzpicture}
 			\caption{Coordinator $M$}
 			\label{Fig:CP1}
 		\end{minipage}
 	
 	\end{minipage}
 	\hfill
 	\centering
 	\begin{minipage}{0.35\textwidth}
 		\centering
 		\begin{minipage}{\textwidth}
 			\centering
 			\begin{tikzpicture}[shorten >=1pt,node distance=2cm,on grid,auto] 
 			
 			\node[state,initial] (q)   { \footnotesize $E||M$}; 
 			\node[state] (q_1) [right=of q] {\footnotesize  $E_0||M$}; 
 			
 			\path[->] 
 			(q)   
 			edge node {\footnotesize $a_0$} (q_1)

 			(q_1) edge  [loop above] node { $a_0$} ();
 			
 			\end{tikzpicture}
 			\caption{$E||M$ satisfies  $\Eventually\Always(\lNot b)$ }
 			\label{Fig:Composition1}
 		\end{minipage}

 	\end{minipage}
 \end{figure}

\paragraph{Example 1: Realizability.}
Consider $E$ defined as follows (Fig~\ref{Fig:Example1}).
\begin{align*}
  & E    =  a_0 \then E_0  \choice a_1 \then E_1 \\
  & E_0  =  a_0 \then E_0 \\
  & E_1  =  b \then E_1
\end{align*}

The self-loop of $b$-actions at $E_1$ violates the specification. Defining $M = a_0 \then M$ (Fig~\ref{Fig:CP1}) ensures that the composed system $E\parallel M$ has a single joint computation, $E\parallel M \trans{a_0} E_0 \parallel M \trans{a_0} E_0 \parallel M \ldots$ (Fig~\ref{Fig:Composition1}), that entirely avoids the state $E_1$. As in Example 0, the structure of $M$ guides the system away from an execution that would violate the specification. 
 
\paragraph{Example 2: Unrealizability.}
Let $E$ be defined as follows.
\begin{align*}
  & E    = a_0 \then E_0 \choice a_0 \then E_1 \\
  & E_0  = a_0 \then E_0 \\
  & E_1  = b \then E_1
\end{align*}
In this case, the coordination problem has no solution. Assume, to the contrary, that $M$ is a solution. Then $M$ must have a transition on $a_0$, otherwise the system $E \parallel M$ deadlocks. Denote the successor of this transition as $M_0$. But then $E\parallel M$ has the infinite joint computation $E \parallel M \trans{a_0} E_1 \parallel M_0 \trans{b} E_1 \parallel M_0 \trans{b} \ldots$ on which $b$ occurs infinitely often, violating the specification.

\paragraph{Fairness for CSP Programs.}
A natural notion of fairness for a group of CSP processes, analogous to strong fairness, considers a computation to be \emph{unfair} if there is a pair of processes that are ready to interact at infinitely many points along the computation, but do so only finitely many times. The assumption of fairness can make a synthesis problem easier to solve, as it limits the set of computations that must be examined for violations of the specification. Unfortunately, the joint computation of $E\parallel M$ constructed in the impossibility proof in Example 2 is fair, as $E$ and $M$ can never interact once the environment is in state $E_1$, which only offers the private action $b$. So the coordination problem in Example 2 has no solution even under fairness.

\paragraph{Example 3: Realizability under fairness.}
Let $E$ be defined as follows.
\begin{align*}
  & E    =  a_0 \then E_0 \choice a_0 \then E_1 \\
  & E_0  =  a_0 \then E_0 \\
  & E_1  =  b \then E_1  \choice a_0 \then E_0
\end{align*}
Let $M = a_0 \then M$. Then $E\parallel M$ has a computation which violates the specification, where the environment enters state $E_1$ and then loops forever on the $b$ action. That computation is unfair, however, as it is possible throughout for $M$ and $E_1$ to synchronize on $a_0$, but such an interaction never occurs. All other computations of $E \parallel M$ satisfy the specification. So the synthesis problem is realizable under fairness, with $M$ as a solution.

\paragraph{Example 4: Partial Knowledge: Non-deterministic Interface.}
Define $E$ as follows.
\begin{align*}
  & E = a_0 \then E_0 \choice a_0 \then E_1 \\
  & E_0 = a_0 \then E_0 \\
  & E_1 = a_1 \then E_1
\end{align*}
Any coordinator must have an initial synchronization on $a_0$, otherwise the system deadlocks. After performing $a_0$, the environment may be either in state $E_0$ or in state $E_1$. In the first state, only action $a_0$ is enabled, while in the second, only $a_1$ is enabled. Hence, a coordinator must be structured to synchronize on \emph{either} of these two actions: if it offers to synchronize only on $a_0$, then if the environment is actually in state $E_1$, the system will deadlock; and vice-versa if it offers only $a_1$. Any synthesis algorithm must resolve such situations where the coordinator has only partial knowledge of the environment state. In this case, the process $M = a_0 \then M_0$ where $M_0 = a_0 \then M_0 \choice a_1 \then M_0$ is a solution.

\paragraph{Example 5: Partial Knowledge: Hidden Actions.}
Consider $E$ defined as follows.
\begin{align*}
  & E = a_0 \then E_0 \\
  & E_0 = b \then E
\end{align*}
This problem is unrealizable. Assume, to the contrary, that $M$ is a solution. To avoid deadlock, $M$ must synchronize on action $a_0$. Let $M_0$ denote the successor state. Note that $E$ synchronously evolves to $E_0$ on $a_0$. To avoid deadlock, it must again be possible to synchronize on $a_0$ from $M_0$. Proceeding in this manner, one  constructs an infinite computation of $E \parallel M$ with infinitely many $a_0$ actions; but this computation must also have infinitely many $b$ actions in-between successive $a_0$ actions, so it violates the specification. A synthesis algorithm must, therefore, track hidden as well as interface actions. 

\section{Background}
\label{Sec:Prelims}

\subsection{Communicating Sequential Processes (CSP)}

CSP has a rich notation and extensive algebraic theory~\cite{DBLP:books/ph/Hoare85,Roscoe-TPC-1997}. As described previously, we use only the basic ``flat'' format of top-level concurrency between processes defined as state machines. A {\em CSP process}, or process (in short)  is defined by a tuple $P = (\State, \Start, \Public, \Private, \delta)$, where
$\State$ is a finite set of states,
$\Start \in \State$ is  a special start state,
$\Public$ are the publicly visible actions  of the process,
and $\Private$ are the privately visible actions of the process. The sets $\Public$ and $\Private$ are disjoint. The transition relation $\delta: \State \times (\Public \cup \Private) \rightarrow 2^\State$ maps each state and action to a set of successor states. A transition from state $s$ on action $a$ to state $t$ exists if $t \in \delta(s,a)$; it is also written as $(s,a,t)$ or sometimes as $s\xrightarrow{a} t$.

A process is {\em deterministic} if for all state $s$ and all actions $ e$, $|\delta(s, e)|\leq 1$, and {\em non-deterministic} otherwise. A public action $a$ is said to be {\em enabled} at state $s$ if there is a transition on the action $a$ from state $s$.


An \emph{execution}, $\pi$, of a process from state $s_0$ is an alternating sequence of 
states and actions $\pi = s_0,a_0,s_1,a_1,\ldots$,
%
%
ending at a state if finite, such that $s_{i+1} \in \delta(s_i,a_i)$ for every $i$ (where $i \geq 0$ and $i < n-1$ if there are $n$ states on the sequence).  
%
%
The sub-sequence $a_0,a_1,\ldots$ of actions is the \emph{trace} of the execution, denoted $\trace(\pi)$ for an execution $\pi$. A \emph{computation} is an execution from the initial state. It is \emph{maximal} if either it is infinite, or it is finite, and no transitions are enabled from the last state of the computation. 
A state $s$ is  \emph{reachable} if there is a finite computation ending at $s$. 

\paragraph{Process Composition and Interaction.}
Let $P$ and $Q$ be CSP processes. Let $X$ be a subset of their common public actions, i.e.,
%
%
$X \subseteq (\Public_P \Inter \Public_Q)$. The composition of $P$ and $Q$
%
%
relative to $X$, denoted $P\parallel_{X} Q$, is a CSP process, with state set $\State_P \times \State_Q$,  initial state $(\Start_P,\Start_Q)$, public actions $(\Public_P \Union \Public_Q) \Diff X$, private actions $(\Private_P \Union \Private_Q \Union X)$, and a transition relation defined by the following rules. 
\begin{itemize}
\item (Synchronized) For an action $a$ in $X$, there is a transition from $(s,t)$ to $(s',t')$ on $a$ if $(s,a,s')$ is a transition in $P$ and $(t,a,t')$ a transition in $Q$. 
\item (Unsynchronized) For an action $b$ in $\Private_P$ or in $\Public_P \Diff X$ (i.e., private, or unsynchronized public action), there is a transition from $(s,t)$ to $(s',t)$ on $b$ if there is a transition $(s,b,s')$ in $P$. A similar rule applies to $Q$.
\end{itemize}

The definition forces processes $P$ and $Q$ to synchronize on actions in $X$; for other actions, the processes may act independently. For simpler notation, we write $\parallel$ instead of $\parallel_{X}$ when $X$ equals $\Public_P \Inter \Public_Q$. 

\paragraph{Fairness.}
There are several ways of defining fairness in CSP (cf.~\cite{DBLP:books/daglib/0067978}). We use the following notion, analogous to strong fairness. Let system $S$ be defined as the parallel composition of a number of processes $S_1, S_2, \ldots, S_n$. A computation of $S$ is \emph{unfair} if there is a pair of processes $S_i,S_j$ (for distinct $i,j$) such that some interaction between $S_i$ and $S_j$ is enabled at infinitely many points along the computation, but the computation contains only finitely many interactions between $S_i$ and $S_j$. Intuitively, a computation is unfair if a system scheduler ignores a pairwise process interaction even if it is infinitely often possible.

\subsection{Linear Temporal Logic} 

We formulate {\em Linear Temporal Logic (LTL)}~\cite{pnueli1977temporal} over an alphabet $\Sigma$, using the syntax~$\varphi ::= \True $ $|\ a \in \Sigma \ |\ \neg \varphi\ |\ \varphi_1\wedge\varphi_2\ |\ \Next \varphi\ |\ \varphi_1 \Untill \varphi_2$.  
The temporal operators are $\Next$(Next) and $\Untill$(Until). The LTL semantics is standard. For an infinite sequence $\pi$ over $\Sigma$, the satisfaction relation, $\pi, i \models \varphi$ for a natural number $i$, is defined as follows:
\begin{itemize}
\item $\pi,i \models \True$ holds for all $i$;
\item $\pi,i \models a \in \Sigma$ if $\pi(i)=a$;
\item $\pi,i \models \neg \varphi$ if $\pi, i \models \varphi$ does not hold;
\item $\pi,i \models \varphi_1 \wedge \varphi_2$ if $\pi, i \models \varphi_1$ and $\pi, i \models \varphi_2$;
\item $\pi,i \models \Next \varphi$ if $\pi, i+1 \models \varphi$; and
\item $\pi,i \models \varphi_1 \Untill \varphi_2$ if there is $j: j \geq i$ such that $\pi, j \models \varphi_2$ and for all $k \in [i,j)$, it is the case that $\pi, k \models \varphi_1$.
\end{itemize}
The Boolean constant $\False$ and  Boolean operators are defined as usual. Other temporal operators are also defined in the standard way: $\Eventually \varphi$ (``Eventually $\varphi$'') is defined as $\True \Untill \varphi$; $\Always \varphi$ (``Always $\varphi$'') is defined as $\neg \Eventually (\neg \varphi)$. For an LTL formula $\varphi$, let $\Lang(\varphi)$ denote the set of infinite sequences $\pi$ over $\Sigma$ such that $\pi,0 \models \varphi$. In the standard formulation of LTL, the alphabet $\Sigma$ is the powerset of a set of ``atomic propositions''.

\subsection{Non-deterministic B\"uchi and Universal co-B\"uchi Automata}
LTL formulas can be turned into equivalent {\em B\"uchi automata}, using standard constructions (e.g.,~\cite{BKRS12}). A {\em B\"uchi automaton}, $A$, is specified by the tuple $(Q,Q_0,\Sigma,\delta,G)$, where $Q$ is a set of states, $Q_0 \subseteq Q$ defines the initial states, $\Sigma$ is the alphabet, $\delta \subseteq Q \times \Sigma \times Q$ is the transition relation, and $G \subseteq Q$ defines the ``green'' (also known as ``accepting'') states. A \emph{run} $r$ of the automaton on an infinite word $\sigma=a_0,a_1,\ldots$ over $\Sigma$ is an infinite sequence $r=q_0,a_0,q_1,a_1,\ldots$ such that $q_0$ is an initial state, and for each $k$, $(q_k,a_k,q_{k+1})$ is in the transition relation. Run $r$ is accepting if a green state appears on it infinitely often; the language of $A$, denoted $\Lang(A)$, is the set of words that have an accepting run in $A$.

A {\em universal co-B\"uchi automaton}, $U$, is also specified by the tuple $(Q, Q_0, \Sigma, \delta, G)$. In this case, green states are also known as ``rejecting" states. A {\em run r} of the automaton on an infinite word $\sigma$ over $\Sigma$ is defined as before. The difference arises in the definition of acceptance: run $r$ is accepting if a green state appears on it \emph{finitely many times}. The language of $U$, denoted $\Lang(U)$, is the set of words for which \emph{all} runs are accepting in this sense.

The complement of the language of a B\"uchi automaton can be viewed as the language of a \emph{universal co-B\"uchi} automaton with the identical structure. Thus, every LTL formula has an equivalent universal co-B\"uchi automaton.

\subsection{Temporal Synthesis}
\label{sec:prelims-synth}

In the standard (synchronous) formulation of temporal synthesis, the goal is to generate a deterministic reactive program $M$ that transforms inputs from domain $\I$ to outputs from domain $\O$. Such a program can be represented as a function $f: \I^{*} \to \O$. The program $M$ (or its functional form $f$) is viewed as a generator of an output sequence in response to an input sequence. For an infinite input sequence $x=x_0,x_1,\ldots$, the generated output sequence $y=y_0,y_1,\ldots$ is defined by: $y_{i} = f(x_0,\ldots,x_{i-1})$, for all $i \geq 0$. (In particular, $y_0$ is the value of $f$ on the empty sequence.) The function $f$ can also be viewed as a strategy for Alice in a turn-based two-player game where Bob chooses the input sequence, while Alice responds at the $i$-th step with the value $y_i$ defined on the history of the play so far. Note that Alice plays first in each step of the game.

Pnueli and Rosner~\cite{pnueli1989synchsynthesis} represent $f$ equivalently as a labeled infinite full-tree, a form that we also use to formulate a solution to the coordination synthesis problem. A full-tree over $\I$ is the set $\I^{*}$. Each finite string over $\I$ represents a node of the tree: the root is $\epsilon$, and for $\sigma \in \I^{*}$ and $a \in \I$, the string $\sigma;a$ is the $a$-successor of the node $\sigma$. A labeling of the tree is a function $\mu: \I^{*} \to \O$ that maps each node of the tree to an output value. It is easy to see that the function $f$ defines a labeling, and vice-versa.

The advantage of this view is that the set of deterministic reactive programs is isomorphic to the set of labeled full-trees. Thus, given an LTL specification $\varphi$ over the alphabet $(\I \times \O)$, Pnueli and Rosner solve the LTL synthesis question (following Rabin~\cite{rabin69}) by constructing a tree automaton that accepts only those labeled trees that represent programs that are solutions, and checking whether this automaton has a non-empty language. A labeled tree represents a solution if it is a full tree and if for every input sequence $x=x_0,x_1,\ldots$, the labels $y=y_0,y_1,\ldots$ on the path in the tree defined by $x$ are such that the sequence of pairs $(x_0,y_0), (x_1,y_1), \ldots$ satisfies $\varphi$. 

The formulation of $M$ (or $f$) here corresponds to a Moore machine; there is an analogous formulation for Mealy machines. A specification is said to be \emph{realizable} if there is a program that satisfies it; it is \emph{unrealizable} otherwise.

\subsection{Solution Methods for Temporal Synthesis}
\label{sec:temporalsynth}
We have already seen the tree-automaton view of the synthesis question.
This has not turned out to be a viable approach for LTL synthesis in practice, however, as the construction of a tree automaton involves complex constructions for determinization and complementation of automata on infinite words. 

In practice, tools such as BoSy~\cite{bosy2017tool} and Acacia+~\cite{bohy2012acacia+} adopt a different strategy called {\em bounded synthesis}~\cite{schewe2007bounded,DBLP:conf/cav/FiliotJR09}. The essential idea is to pick a bound, $N$, on either the number of states in a solution, or the number of visits to a rejecting state in a universal co-B\"uchi automaton. The tool Acacia+ uses this bound to formulate the problem as a \emph{safety} game, and searches for a winning strategy for Alice. In BoSy, the search is for a transition relation for $M$ that has at most $N$ states.

We summarize the BoSy search here, as we use BoSy as the back-end synthesis tool in our implementation. Given a universal co-B\"uchi automaton representing the LTL specification $\varphi$ over alphabet $\I \times \O$ and state space $Q$, and a bound $N: N > 0$ on the number of states of $M$, the search is for a deterministic Moore machine with state space $S=\{0\ldots(N-1)\}$. The synthesis question reduces to whether there exist (1) a transition function $T: S \times \I \to S$, (2) an output function $O: S \to \O$, (3) an inductive invariant $\theta \subseteq Q \times S$, and (4) a rank function $\rho: Q \times S \to [N]$, such that these objects satisfy the standard deductive verification proof rule for universal automata~\cite{DBLP:conf/popl/MannaP87}.

As described in~\cite{FFRT17}, these constraints can be encoded and solved in various ways, either as a propositional satisfiability (SAT) problem, or as a quantified boolean constraint (QBF) problem. The bounded synthesis approach is complete, as there is a known (worst-case exponential) limit on the size of a model for a specification automaton. 


\section{Problem Formulation}
\label{Sec:Model}

\subsubsection*{Specifications}
%
%

The maximal computations of a CSP process may either be finite, ending in a state with no successor, or are infinite. The semantics of LTL is defined only over infinite sequences of actions, but a correctness specification should accommodate both types. To do so, we define a specification $\varphi$ over an action alphabet, say $\kappa$, as a pair $(\varphi_S,\varphi_L)$, where $\varphi_S$ is a set of finite sequences over $\kappa$, specified, say, as a regular expression, and $\varphi_L$ is a set of infinite sequences over $\kappa$, specified in LTL. 

A maximal computation of a CSP process satisfies $\varphi$ in one of two ways:
\begin{itemize}
\item The projection $\pi$ of its trace on $\kappa$ is finite and is in $\varphi_S$, or
\item The projection $\pi$ of its trace on $\kappa$ is infinite and belongs to $\varphi_L$.
\end{itemize}
Note that an infinite computation could have a trace whose projection on $\kappa$ is finite.

To formulate a synthesis algorithm, we assume that the \emph{complement} sets of $\varphi_S$ and $\varphi_L$ are definable by finite automata over $\kappa$. For $\varphi_S$, this is a finite automaton over finite words, while for $\varphi_L$, this is a finite B\"uchi automaton over infinite words in $\kappa$.

\subsubsection*{Coordination Synthesis}
The synthesis problem has been explained informally in prior sections; it is defined precisely as follows. 

\begin{definition}[Coordination Synthesis Problem]
Given an environment process $E = (\State, \Start, \Public, \Private, \delta)$ and a specification $\varphi$ over actions in $(\Public \Union \Private)$, construct a process $M$ (if one exists) with public action set $\Public$ such that all of the maximal finite computations of $E \parallel_{\Public} M$ satisfy $\varphi_S$ and all of its infinite \emph{fair} computations satisfy $\varphi_L$. 
\end{definition}

A problem instance $(E,\varphi)$ is \emph{realizable} if there is a process $M$ meeting the requirements stated above; it is \emph{unrealizable} otherwise.

A process $M$ is {\em non-blocking} for process $E$ if all maximal computations of $E \parallel M$ are infinite; i.e., the composition is free of deadlock. By setting $\varphi_S$ to the empty set, every solution to the problem instance $(E,\varphi)$ is non-blocking, as no maximal finite trace can satisfy $\varphi_S=\emptyset$.

\subsubsection*{Restricting the solution space}

This problem formulation allows a solution $M$ to be non-deterministic and have its own set of internal actions, distinct from $\Public$ and $\Private$
(although the specification cannot, of course, mention those actions). A CSP process can exhibit two types of non-determinism: \emph{internal nondeterminism}, where a state has more than one possible successor on a private action, and \emph{external nondeterminism}, where a state has more than one possible successor on a public action. Intuitively, the first defines a choice that is always enabled and is to be resolved by the process itself, while the second is a choice that is enabled only through synchronization with an external process.

The theorem below shows that there is no loss of generality in restricting the search to deterministic processes without internal actions. This is a simple but important observation, as the behavior of a deterministic process can be viewed as a tree, and automata-theoretic operations on such trees are the basis of our solution to the coordination synthesis problem, as described in the following section.


%

\begin{theorem} \label{thm:deterministic} 
  A synthesis instance $(E,\varphi)$ is realizable if, and only if, it has a deterministic solution process that has no internal actions. 
\end{theorem}

The full proof is in~\cite{coordination-synthesis-tech-report}. Given a solution $M$, the proof constructs a solution $M''$ that has the required properties. It proceeds in two stages. In the first stage, a solution $M'$ is constructed from $M$ by hiding any internal actions of $M$ -- i.e., $(s,a,t)$ is a transition of $M'$ if, and only if, there is a path in $M$ with the trace $\beta;a;\beta'$ from $s$ to $t$, where $\beta, \beta'$ are sequences of internal actions of $M$. As the internal actions of $M$ cannot be referenced in $\varphi$, this is a valid reduction. In the second stage, external non-determinism is removed from $M'$. This could be done by determinizing $M'$ using the standard subset construction, but we also provide a simpler construction that does not incur the worst-case exponential blowup; it merely restricts transitions so that at any state $s$ and any action $a$, at most one of the transitions of $M'$ on $a$ from $s$ is retained. By construction, the resulting $M''$ is externally deterministic and has no internal actions. As $M''$ is simulated by $M'$ while preserving enablement of transitions, $E \parallel M''$ satisfies $\varphi$. 


\section{Tree Views and Automaton-Theoretic Synthesis\label{sec:treeview}}

This section formulates the automaton-based synthesis procedure. As described in Section~\ref{sec:prelims-synth}, the idea is to formulate conditions under which a labeled fulltree defines a solution to the synthesis question. However, one cannot simply reuse the results of Pnueli and Rosner, as the coordination synthesis question has a distinct formulation from the input-output form they consider. The theorems in this section pin down the conditions for validity and show how to represent them as automata.

In a nutshell, the technical development proceeds as follows.
\begin{enumerate}
\item The consequence of Theorem~\ref{thm:deterministic} is that it suffices to limit a solution $M$ to be externally deterministic and free from internal actions. The semantics of such a process can be represented as an infinite labeled fulltree, where each node is labeled with a \emph{set} of public actions and each edge with a public action. 

\item  We show that a labeled fulltree is \emph{invalid} --i.e., it does not represent a solution $M$-- if and only if it has an infinite path satisfying a linear-time property derived from the specification $\varphi$ and the environment process $E$. 

\item We show that this failure property may be represented by a non-deterministic B\"uchi automaton, $\mathcal{B}$, that is (roughly) the product of $E$ with an automaton, $\mathcal{A}$, for the \emph{negation} of $\varphi$. Partial knowledge and asynchrony are both handled in this construction.

\item The \emph{valid} fulltrees are, therefore, those where every path in the tree is accepted by the complement of $\mathcal{B}$. We complement $\mathcal{B}$ implicitly, avoiding a subset construction, by viewing the same structure as a \emph{universal} complemented-B\"uchi automaton. The automaton $\mathcal{B}$ has a number of states that is \emph{linear} in the number of states of $E$ and of $\mathcal{A}$. Its transition relation, however, is of size \emph{exponential} in the number of public interface actions.

\item We give a fully symbolic construction of $\mathcal{B}$ to ameliorate the exponential blowup in the size of its transition relation. The symbolic $\mathcal{B}$ is in a form that can be solved by a number of temporal synthesis tools. The constructed solution (if one exists) is in the form of a Moore machine. We show how to transform this back to a CSP process representing the coordinator $M$ (Section~\ref{Sec:Symbolic}).

\item We prove a complexity hardness result on the scalability of the synthesis problem, showing that it is at least \PSPACE-hard in the size of the environment. It is well known that the synthesis problem is \TWOEXPTIME-complete in the size of the LTL specification~\cite{pnueli1989synchsynthesis}. 
\end{enumerate}

The first few steps explain the issues in terms of fulltrees, which we find easier to follow; the final synthesis procedure \emph{is not based} on tree-automaton constructions, instead it uses co-B\"uchi word-automaton constructions, which are supported by current tools.

\subsection{Labeled Fulltrees and Deterministic CSP Processes}

A \emph{tree} $t$ over alphabet $\Sigma$ is a prefix-closed subset of $\Sigma^{*}$. A \emph{node} of tree $t$ is an element of $\Sigma^{*}$. For a set $\cL$ of \emph{labels}, a $\cL$-labeled tree is a pair $(t,\mu)$ where $\mu: t \to \cL$. I.e., $\mu$ assigns to each node of the tree an element of $\cL$, its label.
%
%
The \emph{full-tree} over $\Sigma$ is the set $\Sigma^{*}$. For a deterministic CSP process $M$, one can generate a full $\Sigma$-tree labeled with $2^{\Sigma}$ as follows. The labeling function, which we denote $\mu_M$, assigns to tree node $\sigma$ the label defined as follows. By determinism, there is at most one computation of $M$ with trace $\sigma$. If this computation exists, let the label $\mu_M(\sigma)$ be the set of actions that are enabled at the state at the end of the computation. Otherwise, let $\mu_M(\sigma)$ be the empty set. We refer to this labeled tree as $\fulltree(M)$.

Conversely, given a full $\Sigma$-tree $(t,\mu)$ labeled with $2^{\Sigma}$, one can extract a deterministic (infinite-state) CSP process, $P=\proc(t,\mu)$, as follows. The state space of $P$ is the set of tree nodes. The initial state is the node $\epsilon$. A triple $(\sigma,a,\delta)$ is in the transition relation of $P$  iff  $a \in \mu(\sigma)$ and $\delta=\sigma;a$. 

%
%

\begin{theorem}
  \label{thm:tree-bisim}
For a deterministic CSP process $M$, the process $M' = \proc(\fulltree(M))$ is bisimular to $M$.
\end{theorem}

\begin{corollary}
  \label{cor:fulltrees}
  The coordination synthesis instance $(E,\varphi)$ is realizable iff there is a full labeled $\Public$-tree $(t,\mu)$, labeled by $2^{\Public}$,  such that $\proc(t,\mu)$ is a solution.
\end{corollary}

\subsection{Recognizing Valid Fulltrees}

We focus on recognizing the kinds of fulltrees defined by Corollary~\ref{cor:fulltrees}. It is actually easier to formulate properties of a fulltree $t$ that \emph{exclude} $\proc(t,\mu)$ from being a solution. From now on, we fix a ``labeled tree'' to mean a $\Public$ fulltree labeled with $2^{\Public}$. 

A \emph{path} in a labeled tree $(t,\mu)$ is an alternating sequence $L_0;a_0;L_1;a_1;\ldots$ where for each $i$, $L_i = \mu(a_0,\ldots,a_{i-1})$ is the label of the node $(a_0,\ldots,a_{i-1})$. If the path is finite, we postulate that it must end with a label. The label at the end of a finite path $\pi$ is referred to as $\Lend(\pi)$. A path is called \emph{consistent} iff $a_i \in L_i$ for all $i$. The trace of a path $\pi=L_0;a_0;L_1;a_1;\ldots$, denoted $\trace(\pi)$, is the sequence of actions on it, i.e., $a_0;a_1;\ldots$.
%
%
Every consistent path $\pi$ in tree $(t,\mu)$ is in correspondence with a computation of $\proc(t,\mu)$. As a computation is a process in itself, we abuse notation slightly and refer to this computation as the process $\proc(\pi)$.


\begin{definition}[Full-tree violation]
%
%
We say that a labeled fulltree $(t,\mu)$ violates $(E,\varphi)$ if one of the following holds.
\begin{enumerate}[(A)]
\item There is a consistent \emph{finite} path $\pi$ in the tree such that some finite computation $\gamma$ in $E \parallel \proc(\pi)$ is maximal and $\trace(\gamma)$ is not in $\varphi_S$, or

\item There is a consistent \emph{finite} path $\pi$ in the tree such that for some infinite computation $\gamma$ of $E \parallel \proc(\pi)$, $\trace(\gamma)$ is not in $\varphi_L$, and no action from $\Lend(\pi)$ is enabled from some point on in $\gamma$, or

\item There is a consistent \emph{infinite} path $\pi$ in the tree such that for some infinite computation $\gamma$ of $E \parallel \proc(\pi)$, $\trace(\gamma)$ is not in $\varphi_L$. 
\end{enumerate}
\end{definition}


\subsubsection*{Illustrative Example}
Consider the CSP environment $E$ defined as follows.
\begin{align*}
  & E = a_0 \then STOP \choice a_0 \then E_0 \choice b_0 \then E_1 \\
  & E_0 = b_1 \then E_0 \\
  & E_1 = a_1 \then E_1
\end{align*}
The environment public actions are $a_0$ and $a_1$, and private actions are $b_0$ and $b_1$. Let $\varphi_S$ be the empty set and $\varphi_L$ be the set of sequences with infinitely many
$a_1$. Figure~\ref{Fig:CandidateTrees} depicts candidate trees for the coordinating process. The composition of tree (a) with $E$ contains a maximal finite computation with trace $a_0$, which 
violates $(E,\varphi)$ by condition (A), as $\varphi_S=\emptyset$. The composition of tree (b) with $E$ results in an infinite computation with trace $a_0 (b_1)^{\omega} \notin \varphi_L$, and therefore violates $(E,\varphi)$ by condition (B), while also violating condition (A), as with tree (a). Finally, the composition of tree (c) with $E$ results in the trace $b_0 (a_1)^{\omega}$, which satisfies $\varphi_L$. Therefore tree (c) is valid for $(E,\varphi)$; it represents the solution $M = a_1\then M$.

\begin{figure}[t]
\centering
\begin{minipage}{0.3\textwidth}
	\centering
	\begin{tikzpicture}[snode/.style={rectangle, draw=black, minimum size=5mm},node distance=1.8cm]
		\node[snode] (q_0)   { {$\{a_0\}$}}; 
		\node[snode] (q_1) [below left of=q_0] {$\cdots$}; 
		\node[snode] (q_2) [below right of=q_0] {$\cdots$}; 
	
		\path[->](q_0)  
			edge node [near end,above=.1cm] {$a_0$} (q_1)
			edge node [near end,above=.1cm] {$a_1$} (q_2);	
	\end{tikzpicture}
	\caption*{(a)}
	\label{Fig:T1}
\end{minipage}
\begin{minipage}{0.3\textwidth}
	\centering
	\begin{tikzpicture}[snode/.style={rectangle, draw=black, minimum size=5mm},node distance=1.8cm]
		\node[snode] (q_0)   { {$\{a_0,a_1\}$}}; 
		\node[snode] (q_1) [below left of=q_0] {$\cdots$}; 
		\node[snode] (q_2) [below right of=q_0] {$\cdots$}; 

		\path[->](q_0)
			edge node [near end,above=.1cm] {$a_0$} (q_1)
			edge node [near end,above=.1cm] {$a_1$} (q_2);	
	\end{tikzpicture}
	\caption*{(b)}
	\label{Fig:T2}
\end{minipage}
\begin{minipage}{0.3\textwidth}
	\centering		
	\begin{tikzpicture}[snode/.style={rectangle, draw=black, minimum size=5mm},node distance=1.8cm]
		\node[snode] (q_0)   { {$\{a_1\}$}}; 
		\node[snode] (q_1) [below left of=q_0] {$\cdots$}; 
		\node[snode] (q_2) [below right of=q_0] {$\{a_1\}$}; 

		\path[->](q_0) 
			edge node [near end,above=.1cm] {$a_0$} (q_1)
			edge node [near end,above=.1cm] {$a_1$} (q_2);	
	\end{tikzpicture}
	\caption*{(c)}
	\label{Fig:T3}
\end{minipage}
\caption{Candidate trees (showing only two-levels of tree). Nodes represent
label, and $\cdots$ indicates "don't care". Tree branches correspond to public
actions.}
\label{Fig:CandidateTrees}
\end{figure}

\begin{theorem}
  Fix a synthesis instance $(E,\varphi$). For any labeled fulltree $(t,\mu)$, the CSP process $\proc(t,\mu)$ is a solution if, and only if, $(t,\mu)$ is not violating for $(E,\varphi)$.
\end{theorem}
\begin{proof}
  We will show the contrapositive.
  
  (left-to-right)
    Suppose that $(t,\mu)$ violates $(E,\varphi)$. Then one of (A)-(C) holds. We show that process $M=\proc(t,\mu)$ cannot be a solution. 

    Suppose case (A) holds. The maximal finite computation $\gamma$ of $E \parallel \proc(\pi)$ is also a maximal computation of $E \parallel M$, as $\pi$ corresponds to a computation of $M=\proc(t,\mu)$. Since $\trace(\gamma)$ is not in $\varphi_S$, $M$ is not a solution. 

    Suppose case (B) holds. The infinite computation $\gamma$ of $E \parallel \proc(\pi)$ is also a computation of $E \parallel M$. As $\pi$ is finite, $\gamma$ has only finitely many synchronization actions. Moreover, none of the actions of $\Lend(\pi)$ are enabled on $\gamma$ from some point on. Thus, no synchronization between $M$ and $E$ is enabled on $\gamma$ from that point on, so the computation $\gamma$ is a fair computation of $E \parallel M$ whose trace is not in $\varphi_L$.

  Suppose case (C) holds. The computation $\gamma$ of $E \parallel \proc(\pi)$ is an infinite computation of $E \parallel M$, which is fair as it has infinitely many synchronizations. However, $\trace(\gamma)$ is not in $\varphi_L$.

  (right-to-left)
  Suppose $M=\proc(t,\mu)$ is not a solution for $(E,\varphi)$. We show that $(t,\mu)$ must violate $(E,\varphi)$. As $E\parallel M$ fails to satisfy $\varphi$, there are two possibilities.

  The first possibility is that of a  maximal finite computation of $E\parallel M$ whose trace is not in $\varphi_S$. This computation corresponds to a finite consistent path $\pi$ in the tree $(t,\mu)$ that meets condition (A).

  The other possibility is that there is an infinite fair computation $\gamma$ of $E\parallel M$ whose trace is not in $\varphi_L$. There are two types of fair computations.

  In the first type, $\gamma$ has only finitely many synchronization actions followed by an infinite suffix where no synchronization is enabled with $M$. As $M$ has no internal actions, the suffix of $\gamma$ consists solely of internal actions of $E$. The shortest finite prefix of $\gamma$ containing all synchronization actions defines a consistent finite path $\pi$ in $(t,\mu)$. The set of enabled actions of $M$ at the end of this prefix of $\gamma$ is (by the definition of $M=\proc(t,\mu)$) the label at the end of $\pi$. Thus, $\pi$ satisfies the requirements of condition (B).

  In the second type, $\gamma$ has infinitely many synchronization actions. This induces an infinite computation of $M$, and a corresponding consistent infinite path $\pi$ of $(t,\mu)$ meeting condition (C). 
%
\end{proof}

\subsection{Automata-Theoretic Synthesis}
\label{Sec:SpecAutomataConstruction}
We now describe how to construct B\"uchi automata that recognizes each of the properties (A)-(C) on paths of a fulltree. From these, one can construct a  \emph{universal} co-B\"uchi automaton that accepts a path of the tree either if it is inconsistent, or if it does not satisfy either of (A)-(C), called the {\em specification automaton}. In the classical setting (cf.~\cite{pnueli1989synchsynthesis}), this universal automaton is lifted to a tree automaton that runs on all paths of a tree. Conjoined with a tree automaton that checks whether the tree is a fulltree, one obtains a tree automaton that accepts the set of valid full-trees. The synthesis problem is realizable iff the language of that automaton is non-empty. In the alternative setting of bounded synthesis~\cite{schewe2013bounded-journal}, this universal co-B\"uchi automaton is used directly in the check for realizability. 

The automata for properties (A)-(C) run on infinite sequences over an input alphabet $\Public \times 2^{\Public}$. An input sequence of the form $(a_0,L_0),(a_1,L_1),\ldots$ represents the labeled fulltree path $L_0;a_0;L_1;a_1;\ldots$. The size of the automaton transition relation is thus exponential in $|\Public|$. Note that the $\Private$ actions do not appear in the input alphabet, they are eliminated in the process of constructing the automaton.

\paragraph{Checking for consistency.}
The first automaton accepts an input sequence of the form above if it corresponds to a consistent path, i.e., if $a_i \in L_i$ for all $i$. This automaton has a constant number of states. 

\paragraph{Checking condition (A)}
Let $\mathcal{A}_S$ be a finite automaton for the negation of $\varphi_S$. The automaton for (A) is a product of $\mathcal{A}_S$ and $E$. It guesses a \emph{finite} computation $\gamma$ of $E$, of the form $e_0;[\beta_0,a_0,\beta'_0];e_1;[\beta_1,a_1,\beta'_1];\ldots;[\beta_n,a_n,\beta'_{n}];e_n$, where the $\beta$ and $\beta'$ symbols represent sequences of private actions of $E$ (the intermediate states of $E$ are not shown, and the grouping in square brackets, e.g., $[\beta_0,a_0,\beta'_0]$, is just for clarity). The automaton also concurrently guesses a run of the finite-word automaton $\mathcal{A}_S$ on $\trace(\gamma)=\beta_0,a_0,\beta'_0,\ldots,\beta_n,a_n,\beta'_n$. It accepts if the automaton is in an accepting state, the final state $e_n$ of $\gamma$ has no transitions on private actions, and none of its enabled public transitions are in $L_n$, the final label. This automaton has $|\mathcal{A}_S|*|E|$ states.

\paragraph{Checking condition (B)} 
Let $\mathcal{A}_L$ be a finite B\"uchi automaton for the negation of $\varphi_L$. The automaton for (B) guesses, as in the construction for case (A), a finite computation $\gamma$ of $E$ of the form $e_0;[\beta_0,a_0,\beta'_0];e_1;[\beta_1,a_1,\beta'_1];\ldots;[\beta_n,a_n,\beta'_{n}];e_n$, where the $\beta$ and $\beta'$ symbols represent sequences of private actions of $E$, and it concurrently also guesses a run of the infinite word automaton $\mathcal{A}_L$ on $\trace(\gamma)=\beta_0,a_0,\beta'_0,\ldots,\beta_n,a_n,\beta'_n$ to an automaton state, say, $q$. It then guesses a further infinite execution $\gamma'$ of $E$, starting at the final state $e_n$ of $\gamma$, such that all transitions on $\gamma'$ are on private actions of $E$, and it checks that at each intermediate state of $\gamma'$, none of the enabled public transitions are in $L_n$. It concurrently simulates an extension to the run of $\mathcal{A}_L$ on $\trace(\gamma')$ from state $q$ and accepts if this run is accepting for $\mathcal{A}_L$. This automaton has  $|\mathcal{A}_L|*|E|$ states.

\paragraph{Checking condition (C)} 
Let $\mathcal{A}_L$ be a finite B\"uchi automaton for the negation of $\varphi_L$. The automaton for (C) guesses an infinite computation $\gamma$ of $E$ of the form $e_0;[\beta_0,a_0,\beta'_0];e_1;[\beta_1,a_1,\beta'_1];\ldots$ on the infinite trace $a_0,a_1,\ldots$, and it concurrently simulates a run of the automaton $\mathcal{A}_L$ on $\trace(\gamma)$ (which contains private actions represented by the $\beta$ and $\beta'$ symbols), accepting if this run accepts. This automaton has $|\mathcal{A}_L|*|E|$ states.

\paragraph{Constructing the Specification automaton.} 
We construct an automaton by taking the union of the automata for (A)-(C) and intersecting the union with the automaton that checks consistency of the input sequence. This is then a non-deterministic B\"uchi automaton which accepts an input sequence if it is consistent and satisfies one of (A)-(C). Its complement co-B\"uchi automaton accepts an input sequence if it is inconsistent or it fails to meet either of (A)-(C). This co-B\"uchi automaton is the {\em specification automaton}. The specification  automaton has a  number of states proportional to $(|\mathcal{A}_S|+|\mathcal{A}_L|)*|E|$. The number of transitions is, however, exponential in the number of public actions, as transitions are labeled with subsets of public actions. 

\subsection{Hardness of Coordination Synthesis}

We show that the coordination synthesis problem is \PSPACE-hard in terms of the size of the environment, $E$, keeping the specification formula fixed.

\begin{theorem}
  Coordination Synthesis is \PSPACE-hard in $|E|$ for a fixed specification. 
\end{theorem}

\begin{proof}
  The proof is by reduction from a standard \PSPACE-complete problem: given a finite automaton, $A$, determine whether the language of $A$ is \emph{not} the universal language over its alphabet, $\Sigma$.  Without loss of generality, we suppose that $A$ is complete -- i.e., each state has a transition on every letter in $\Sigma$.

  From $A$, we construct a synthesis instance $(E,\varphi)$ as follows. The environment $E$ is a copy of $A$, with two additional states, $\acc$ and $\rej$, and three additional letters, $\sharp$, $+$, and $-$ . Its transition relation is that of $A$, together with the following additional transitions. First, from each accepting state of $A$, there is a transition on $\sharp$ to the $\acc$ state. The $\acc$ state has a self-loop transition on $+$. Second, from each non-accepting state of $A$, there is a transition on $\sharp$ to the $\rej$ state. The $\rej$ state has a self-loop transition on letter $-$. Note that $E$ has no internal actions but, in general, will have non-deterministic choice on interface actions.

  The specification $\varphi$ is given by $\varphi_S = \emptyset$ (i.e., non-blocking solutions are required) and $\varphi_L$ being the LTL formula  $\Eventually(\sharp \lAnd \Next\Always(-)))$. We show that the language of $A$ is non-universal if, and only if, the synthesis problem $(E,\varphi)$ has a solution. 

  In the left-to-right direction, suppose that $A$ has a non-universal language. Let $w$ be a word that is not accepted by $A$. Let $M$ be the deterministic process with a single execution with trace $w\sharp(-)^{\omega}$. We claim that $M$ is a solution for $(E,\varphi)$. As $A$ is complete, there can be no maximal computation of $E\parallel M$ with a trace that is a strict prefix of $w$. And as every run of $w$ on $A$ ends in a non-accepting state, every maximal computation of $E \parallel M$ is infinite and has the trace $w\sharp(-)^{\omega}$, which meets the specification. Hence, $M$ is a solution to the synthesis problem. 
  
  In right-to-left direction, let $M$ be any solution to the synthesis problem. Thus, $M$ is non-blocking, every maximal computation of $E\parallel M$ is infinite and its trace satisfies $\varphi_L$. From Theorem~\ref{thm:deterministic}, we may suppose that $M$ is deterministic and has no internal actions. As $M$ is non-blocking and $E$ has no dead-end states by construction, there is at least one infinite joint computation, denoted $x$, of $E \parallel M$. As $E\parallel M$ satisfies $\varphi_L$, the trace of $x$ must have the pattern $w\sharp(-)^{\omega}$, for some $w \in \Sigma^{*}$. We show that $w$ is rejected by $A$. Let $s$ be the state of $M$ in $x$ following the prefix with trace $w$. Then $s$ has a transition on $\sharp$.  

  Suppose, to the contrary, that there is an accepting run of $A$ on $w$ ending in some state $t$. Thus, there is a joint computation of $E\parallel M$ that follows this run and ends at the joint state $(t,s)$. As both $t$ and $s$ have transitions on $\sharp$, this computation can be extended to $(\acc,s')$ on $\sharp$, for $s'$. As $\acc$ only has transitions on $+$, and as $M$ is non-blocking, this computation must extend further to a computation with trace $w\sharp(+)^{\omega}$. But that computation does not satisfy $\varphi_L$, a contradiction.
  %
\end{proof}

\section{Symbolic Constructions\label{Sec:Symbolic}}
In this section, we give a fully symbolic construction for the automaton of the transformed specification.
The symbolic form ameliorates the exponential blowup in the size of its transition relation. The structure of the automaton is defined as follows. 
\begin{itemize}
\item The states are either (1) special states $\Fail$ and $\Sink$, or (2) normal states of the form $(q,r,e)$, where $q$ is a state of $\mathcal{A}_L$,  $r$ is a state of $\mathcal{A}_S$, and $e$ is an environment state.

\item The initial state is $(q_0,r_0,e_0)$ where all components are initial in their respective structures.

%
%
\item The input symbols have the form $(a,L,g)$, where $a$ is an public action, $L$ is a set of public actions, and $g$ is a Boolean marking the transition as being ``green'' (accepting) if true. 
  
\item The transition relation of the joint automaton is described as a predicate, $T(\mathsf{current\_state},$ $\mathsf{input\_symbol},$ $\mathsf{next\_state})$. On all input symbols, the special states $\Fail$ and $\Sink$ have a self loop; i.e., the next state is the same as the current state. The transition relation for normal states is defined below.

\item The green (accepting) edges of the automaton are all transitions where $g$ is true.

\item The green (accepting) states of the automaton are the special state $\Fail$ and every state $(q,r,e)$ where $q$ is a green state of $\mathcal{A}_L$.
\end{itemize}

For a normal state $(q,r,e)$ and input symbol $(a,L,g)$, the transition relation $T$ is defined as follows. It relies on several auxiliary relations which are defined below. 

\begin{enumerate}
\item (Condition (A)) If $\Efail(r,e,L)$ holds,  the successor state is $\Fail$.
\item (Condition (B)) Otherwise, if $\noSynch(q,r,e,L)$ holds, the successor state is $\Fail$.
\item (Non-blocking) Otherwise, if $\Esink(a,e,L)$ holds, the successor state is $\Sink$.
\item (Condition (C)) Otherwise,  if $\normtrans((q,r,e), (a,g), (q',r',e')) $ holds, the successor is the normal state $(q',r',e')$.
\end{enumerate}


The auxiliary relations represent the existence of paths, and are defined in the form of a least fixpoint. Each of these fixpoint computations, as well as the construction of the transition relation, can be computed symbolically using BDDs. We provide the fixpoint formulations below. In these formulations, for an action $c$ (public or private), the joint transition $\Joint((q,r,e),c,(q',r',e'))$ holds iff $\mathcal{A}_L(q,c,q')$, $\mathcal{A}_S(r,c,r')$, and $E(e,c,e')$ all hold. 

Following condition (A), predicate $\Efail(r,e,L)$ holds if there is a path consisting only of private transitions from state $e$ in $E$ to an end-state $e'$ that fails the $\varphi_S$ property. The state $e'$ can have no private transitions and no enabled public actions in $L$. Being a reachability property, the predicate is defined as the least fixpoint of $Y(r,e,L)$ where
\begin{itemize}
\item (Base case) If $r$ is a final state of $\mathcal{A}_S$, $e$ has no private transitions, and none of its enabled public actions is in $L$, then $Y(r,e,L)$ is true, and
\item (Induction) If there exists $r',e'$ and a private action $b$ such that $Y(r',e',L)$, $\mathcal{A}_S(r,b,r')$ and $E(e,b,e')$ hold, then $Y(r,e,L)$ holds.
\end{itemize}


Following condition (B), predicate $\noSynch(q,r,e,L)$ holds if there exists a infinite path consisting of only private transitions that fails the $\varphi_S$ property, and from some point on $L$ does not synchronize with the public actions enabled on the path. 
One can easily argue that it is  sufficient to find an {\em infinite lasso path} consisting  of only private transitions that fails the $\varphi_S$ property, and no state in the {\em loop of the lasso} synchronizes with $L$. Formally, $\noSynch(q,r,e,L)$ holds if there exists a private action $b$, states $q_0,r_0,e_0,g_0$ and $q_1,r_1,e_1,g_1$ such that all of the the following conditions hold
\begin{itemize}
	\item $\Eprivate((q,r,e),g_0,(q_0,r_0,e_0))$, 
	\item $\Joint((q_0,r_0,e_0),b,(q_1,r_1,e_1))$ 
	\item $\geneprivate((q_1,r_1,e_1),g_1,L, (q_0,r_0,e_0))$, and 
	\item $g \lEquiv g_1$. 
\end{itemize}

The predicate $\Eprivate((q,r,e),g,(q',r',e'))$ holds if there is a path consisting only of private transitions from state $e$ in $E$ to $e'$, a run of the $\mathcal{A}_S$ automaton on this path from state $r$ to $r'$, a run of the $\mathcal{A}_L$ automaton on the same path from state $q$ to $q'$, and $g$ is true if one of the states on the $\mathcal{A}_L$ automaton run is a green (i.e., B\"uchi accepting) state.  Being a reachability property, this can be defined as the least fixpoint of $Z((q,r,e),g,(q',r',e'))$, where
\begin{itemize}
	\item (Base case) If $q=q',r=r',e=e'$ then $Z((q,r,e),g,(q',r',e'))$ holds, and $g$ is true iff $q$ is green, and
	\item (Induction) If $Z((q,r,e),g_0,(q_0,r_0,e_0))$ and $J((q_0,r_0,e_0),b,(q',r',e'))$ for a private action $b$, then   $Z((q,r,e),g,(q',r',e'))$ holds, with $g$ being true if $g_0$ is true or $q'$ is green.
\end{itemize}

Predicate $\geneprivate((q,r,e), g, L, (q',r',e'))$  generalizes
$\Eprivate((q,r,e),g,(q',r',e'))$ by additionally ensuring that no state along the run from $e$ to $e'$ in $E$ enables a public action in $L$.
Predicate $\Esink(a,e,L)$ holds if either $a \not \in L$ (inconsistent), or $\lNot \Enabled(a,e)$ (blocking) where
predicate $\Enabled(a,e)$ holds if there is a path in $E$ consisting only of private transitions from state $e$ to a state from which there is a transition on public symbol $a$. Being a reachability property, this predicate is defined as the least fixpoint of $X(a,e)$ where
\begin{itemize}
	\item (Base case) If $e$ has a transition on $a$, then $X(a,e)$ holds.
	\item (Induction) If there exists $e'$ and a private action $b$ such that $E(e,b,e')$ and $X(a,e')$ hold, then $X(a,e)$ holds.
\end{itemize}

  Following condition (C), predicate $\normtrans((q,r,e), (a,g), (q',r',e')) $ holds if 
  there is a path of private actions to the normal target state $(q',r',e')$ with a single public transition on $a$, along with  matching runs of the two automata that fails the $\varphi_S$ property.  Formally, $\normtrans((q,r,e), (a,g), (q',r',e')) $ holds if there exists  states $q_0,r_0,e_0,g_0$ and $q_1,r_1,e_1,g_1$ such that all of the the following conditions hold
  \begin{itemize}
  	\item $\Eprivate((q,r,e),g_0,(q_0,r_0,e_0))$, 
  	\item $\Joint((q_0,r_0,e_0),a,(q_1,r_1,e_1))$, 
  	\item $\Eprivate((q_1,r_1,e_1),g_1,(q',r',e'))$, and 
  	\item $g \lEquiv g_0 \lOr g_1$. 
  \end{itemize}

\subsubsection*{Encoding into Pnueli-Rosner synthesis}
The resulting automaton has transitions on $(a,L,g)$. The green transitions can be converted to green states by a simple construction (omitted) that adds a copy of the state space. This automaton operates on input sequences of the form $(a_0,L_0),(a_1,L_1),\ldots$. That matches the form of the \emph{synchronous} Pnueli-Rosner model where automata operate on input sequences of the form $(x_0,y_0),(x_1,y_1),\ldots$, where the $x$'s represent input values and the $y$'s the output value chosen by the synthesized process. By reinterpreting the $a$'s as $x$'s and $L$'s as $y$'s, one can use existing tools for synchronous synthesis, to check for realizability and produce a solution if realizable. In this re-interpretation, if the public alphabet has $n$ symbols, there are $\mathsf{log}(n)$ bits for $x$ (i.e., $a$), and $n$ bits for $y$ (i.e., $L$). 

\subsubsection*{Mapping a solution back to CSP}
The synchronous synthesis tools  produce a deterministic Moore machine as a solution to the reinterpreted synthesis problem. Such a machine maps a sequence of $x$-values (i.e., a sequence of public actions, reinterpreted) to a $y$-value (i.e., a set of public actions, reinterpreted). That is precisely the CSP process $M$, defined as follows. The states of $M$ are the states of the Moore machine. There is a transition $(s,a,t)$ in $M$ if, and only if, $a$ is included in the output label $y$ of the Moore machine at state $s$, and the machine has a transition from $s$ to $t$ on input $x=a$.

\section{Implementation and Case Studies}
\label{sec:implementation}

\subsection{Implementation and Workflow}
We implement the coordination synthesis algorithm in about $1500$ lines of Python code. The implementation uses the \textsf{FDR4} CSP tool~\cite{DBLP:journals/sttt/Gibson-Robinson16,fdr4} to read in a CSP description of the environment written in the language $CSP_M$, and produce a flattened description in the form of a state machine for the environment $E$. It uses features of the \textsf{SPOT} system~\cite{spot,lutzspot} to read in LTL formulas that define the specifications $\varphi_S$ and $\varphi_L$, and to convert those to the appropriate automaton form. The core of the construction (about $1000$ lines of Python) implements the symbolic specification automaton construction described in Sections~\ref{sec:treeview} and~\ref{Sec:Symbolic}. The resulting automaton, $\mathcal{B}$, is simplified using the SPOT tool \textsf{autfilt} to, for instance, remove dead states. It is then supplied to the synchronous synthesis engine \textsf{BoSy}~\cite{Bosy,FFRT17}. BoSy produces a Moore machine, which is converted into a CSP process as described at the end of Section~\ref{Sec:Symbolic}. We are grateful to the authors of these and supporting tools for making the tools freely available.

As defined in Section~\ref{sec:treeview}, the automaton $\mathcal{B}$ operates on inputs from the alphabet $\Sigma \times 2^{\Sigma}$, where $\Sigma$ is the set of public interface actions of the environment $E$. Thus, an input symbol has the form $(a,L)$, where $a \in \Sigma$ and $L$ is a subset of $\Sigma$. BoSy expects these values to be encoded as bit vectors defining the input, denoted $x$, and output, denoted $y$, of the synchronous synthesis problem. For simplicity, we choose a 1-hot encoding for the symbol $a$; i.e., we have a bit $x_m$ for each symbol $m$ in $\Sigma$, and ensure that exactly one of those bits is true -- to represent symbol $a$, the bit $x_a$ is set to true while all others are set to false. To encode the set $L$, we define bits $y_m$ for each symbol $m$ in $\Sigma$. Thus, a set $L$ is encoded by setting $y_m$ to be true for all $m$ in $L$, and $y_m$ to false for all $m$ not in $L$.

Recall from Section~\ref{Sec:Prelims} that BoSy encodes the synchronous synthesis problem as constraint solving, either propositional (with SAT solvers) or quantified (with QBF solvers). The QBF form has the quantifier prefix $\exists\forall\exists$, where the universal quantification is over the input variables; in effect, the SAT encoding is obtained by replacing the $\forall$ quantifier with a conjunction over all values of the input variables. In our case, the input variables represent elements of $\Sigma$, as explained above. As this set is not very large (for our examples), we can replace the $\forall$ quantifier with a conjunction over all elements of $\Sigma$, and obtain an equivalent SAT problem. Doing this in practice required a small modification to BoSy, as its default is to expand over \emph{all} assignments  to input variables, which would lead to a $|2^{\Sigma}|$ blowup in the expansion for our 1-hot encoding. Instead, we modify BoSy to expand over only 1-hot assignments to the input variables, which gives the expected linear $|\Sigma|$ blowup.
We refer to our SAT encoding by OneHot-SAT.

In prior work on BoSy~\cite{FFRT17} the authors report that QBF outperforms SAT. Much to our surprise, we observed the reverse: the SAT encoding handily outperforms QBF, as shown in the experiment tables later in this section. We use the
state-of-the-art SAT solver \textsf{CryptominiSAT}~\cite{DBLP:conf/sat/SoosNC09}.

All experiments were run on 8 CPU cores at 2.4GHz, 16GB RAM, running 64-bit Linux.

\subsection{Case studies}

The objective of our case studies is to demonstrate the utility of coordination synthesis in designing coordination programs under various aspects of asynchrony, partial information and concurrency. We begin by testing our implementation on examples from Section~\ref{Sec:MotivatingEg} (Section~\ref{Sec:CaseStudy-Examples}), and then follow it up by designing a coordination program for a {\em smart thermostat} (Section~\ref{Sec:CaseStudy-Thermostat}) and an {\em arbiter} for concurrent processes (Section~\ref{Sec:CaseStudy-Arbiter}). We explore different aspects of the applicability of coordination synthesis in each of these studies, ranging from modeling choices to fairness in concurrency. Finally, we discuss the scalability challenges in Section~\ref{Sec:CaseStudy-Verdict}.

\subsubsection{Illustrative examples from \textsection~\ref{Sec:MotivatingEg}}
\label{Sec:CaseStudy-Examples}
We begin with demonstrating the correctness of our implementation by synthesizing all specifications from the illustrative examples given in Section~\ref{Sec:MotivatingEg}. This set of examples is representative of the major features and complexities of the program model that our synthesis procedure must account for, and hence make for a test bench with good coverage.

Each input specification consists of an environment process, a safety LTL specification and a liveness LTL specification. In this set of inputs, the safety and liveness specifications are fixed to $ \mathsf{False}$ and $\Eventually\Always(\neg b)$, respectively.
Each input instance is run twice, once when the synchronous synthesis tool BoSy invokes its default QBF encoding, and once when BoSy invokes our domain-specific OneHot-SAT encoding. If realizable, we output the synthesized coordinator as a CSP process. 
Our observations are summarized in Table~\ref{tab:illustrativeEg}.
Our main observations are:
\begin{enumerate}
\item Our implementation returned the \textbf{expected} outcomes for  realizable specifications. The tool times out on unrealizable specifications, which is expected as the bound $N$ (Section~\ref{sec:temporalsynth}) for BoSy to guarantee unrealizability is very large. 

\item For all but one example  (Example 4) the coordination programs obtained from both the encoding options of BoSy were identical to the ones constructed by hand in Section~\ref{Sec:MotivatingEg}.

In Example 4, both encodings returned the same coordinator, but different from the one in Section~\ref{Sec:MotivatingEg}. The synthesized program  is \textbf{ smaller } than the ones made by hand.

\end{enumerate}

These observations indicate that our implementation is faithful to the theoretical development of coordination synthesis.

\begin{table}[t]
\centering
\caption{Evaluation of implementation on Illustrative examples. Safety spec is $\mathsf{False}$, liveness spec is $\Eventually\Always (\neg b)$.}
\label{tab:illustrativeEg}
\begin{tabular}{|c|c|c|c|c|}
\hline
\multirow{3}{*}{CSP} & \multirow{3}{*}{Synthesized coordinator} & \multicolumn{3}{c|}{Run time (in milliseconds, timeout=10s)}                                  \\ \cline{3-5} 
                     &                                          & \multirow{2}{*}{Spec. aut. construction} & \multicolumn{2}{c|}{Sync. synthesis} \\ \cline{4-5} 
                     &                                          &                                          & QBF (Default)           & OneHot-SAT         \\ \hline
Example 0            & $M = a_0 \then M$                        & 138                                      & 955             & 86                 \\ \hline
Example 1            & $M = a_0 \then M$                        & 135                                      & 1082            & 89                 \\ \hline
Example 2            & \textsf{None}                             & 137                                      & Timeout         & Timeout            \\ \hline
Example 3            & $M = a_0 \then M$                        & 136                                      & 798             & 86                 \\ \hline
Example 4            & $M = a0 \then M \choice a1\then M$       & 138                                      & 773             & 77                 \\ \hline
Example 5            & \textsf{None}                             & 138                                      & Timeout         & Timeout            \\ \hline
\end{tabular}
\end{table}

\subsubsection{Thermostat: Case study on iterative development}
\label{Sec:CaseStudy-Thermostat}
 
 \begin{figure}[t]
 	\centering
 	\begin{minipage}{0.55\textwidth}
 		\centering
 		\begin{tikzpicture}[shorten >=1pt,node distance=2.5cm,on grid,auto] 
 		
 		\node[state,initial] (q_0)   { \footnotesize \textsf{JR}}; 
 		\node[state] (q_2) [right=of q_0] {\footnotesize\textsf{TW}}; 
 		\node[state] (q_1) [left=of q_0] {\footnotesize \textsf{TC}}; 
 		
 		\path[->] 
 		(q_0)   edge [loop above] node { \footnotesize \textsf{JustRight}} ()
 		edge [bend left = 30, red] node {\footnotesize \textsf{HeatisOn}} (q_2)
 		edge [bend right = 10,dashed] node {} (q_2)
 		edge [bend right = 30, blue] node [above] {\footnotesize \textsf{ACisOn}} (q_1)
 		edge [bend right = 10,dashed] node  {} (q_1)	
 		
 		(q_1) 
 		edge [bend right = 30, red] node [sloped,below] {\footnotesize \textsf{HeatisOn}} (q_0)
 		edge [bend right = 10,dashed] node {\footnotesize \textsf{}} (q_0)
 		edge [loop left, blue] node { \footnotesize \textsf{ACisOn}} ()
 		edge [loop above] node { \footnotesize \textsf{Cold}} ()
 		
 		(q_2) 	edge [bend left = 30, blue] node  {\footnotesize \textsf{ACisOn}} (q_0)
 		edge [bend right = 10,dashed] node {\footnotesize \textsf{}} (q_0)
 		edge [loop right, red] node { \footnotesize \textsf{HeatisOn}} ()
 		edge [loop above] node { \footnotesize \textsf{Warm}} ();
 		
 		\end{tikzpicture}
 		\caption{Room temperature $\mathsf{Sensor}$ CSP (JR: Just right, TW: Too warm, TC: Too cold)}
 		\label{Fig:Sensor}
 	\end{minipage}
 	\hfill
 	\centering
 	\begin{minipage}{0.4\textwidth}
 		\centering
 		\begin{minipage}{\textwidth}
 			\centering
 			\begin{tikzpicture}[shorten >=1pt,node distance=2.35cm,on grid,auto] 
 			\node[state,initial] (q_0)   { \footnotesize \textsf{OFF}}; 
 			\node[state] (q_1) [right=of q_0] {\footnotesize \textsf{ON}}; 
 			\path[->] 
 			(q_0)   
 			edge [bend left = 20]node { \footnotesize \textsf{switchHeatOn}}  (q_1)
 			(q_1) edge  [loop right, red] node  { \footnotesize \textsf{HeatisOn}} ()
 			edge [bend left = 20]node {\footnotesize \textsf{switchHeatOff}} (q_0);
 			
 			\end{tikzpicture}
 			\caption{$\mathsf{Heater}$ CSP}
 			\label{Fig:Heater}
 		\end{minipage}
 		\begin{minipage}{\textwidth}
 			\centering
 			\begin{tikzpicture}[shorten >=1pt,node distance=2.35cm,on grid,auto] 
 			\node[state,initial] (q_0)   { \footnotesize \textsf{OFF}}; 
 			\node[state] (q_1) [right=of q_0] {\footnotesize \textsf{ON}}; 
 			\path[->] 
 			(q_0)   
 			edge [bend left = 20]node { \footnotesize \textsf{switchACOn}}  (q_1)
 			(q_1) edge  [loop right, blue] node  { \footnotesize \textsf{ACisOn}} ()
 			edge [bend left = 20]node {\footnotesize \textsf{switchACOff}} (q_0);

 			\end{tikzpicture}
 			\caption{$\mathsf{AC}$ CSP}
 			\label{Fig:AC}
 		\end{minipage}
 		
 	\end{minipage}
 \end{figure}


 Synthesis can be thought of as a declarative paradigm for programming where a developer can focus on the ``what''s of the coordination problem rather than on the ``how"s. Through this case study we illustrate how  synthesis simplifies the development of coordination programs. Our objective is to design a {\em smart thermostat}. 
 
 A smart thermostat interacts with a room-temperature sensor  (sensor, in short), 
 a heater, and an air-conditioner in order to maintain a comfortable room temperature. The temperature can be affected by the mode (switch-on or switch-off) of
 the heater and air-conditioner, and by external physical factors such as weather  fluctuations, which are unpredictable and cannot be controlled. These factors prevent the smart thermostat from assessing the room temperature correctly from the modes of the devices alone. As a result, the smart thermostat must communicate with the sensor to check the room temperature, and respond accordingly to maintain ambient temperature.

 CSP processes modeling the  sensor, heater and air-conditioner are given in Figures
 \ref{Fig:Sensor}-\ref{Fig:AC}. The states of the sensor denote 
 the current temperature,
 while states of the heater and air-conditioner denote their mode.
 The dashed-transitions in the sensor model fluctuations in room temperature caused by changing external physical conditions and cause internal actions. 
 The actions \textsf{HeatIsOn} (red transitions) and \textsf{AcIsOn} (blue transitions)  cause private agent-to-agent interactions between the sensor and
 the heater or air-conditioner, respectively.  
 Finally, the sensor communicates the current room temperature to the
 smart thermostat through actions \textsf{Cold}, \textsf{JustRight},
 \textsf{Warm}, and the heater and air-conditioner interact with the
 smart thermostat through the \textsf{Switch} actions (black transitions).
 Therefore, the flat environment in represented by
 $$ \mathsf{ENV} = \mathsf{Heater} ||_{\{\mathsf{HeatisOn}\}} \mathsf{Sensor} ||_{\{\mathsf{ACisOn}\}} \mathsf{AC}$$
 Its interface with the coordinator consists of public actions $\{\mathsf{JustRight}, \mathsf{Cold}, \mathsf{Warm}, \mathsf{switchACOff}$, $\mathsf{switchACOn}, \mathsf{switchHeatOff}, \mathsf{switchHeatOn}\}$.
 
 The safety specification is simply deadlock freedom, given by the formula $\mathsf{False}$.
 The goal is to maintain ambient temperature. So, we begin with the following liveness specification, asserting that the temperature is infinitely often just right: $$\mathsf{AmbientTemp} := \Always\Eventually (\mathsf{JustRight})$$
 
 This coordination synthesis specification was satisfied by the trivial coordinator $$M = \mathsf{JustRight} \then M$$
 
 Clearly, this coordinator relies on the internal actions of the sensor (dashed transitions) to maintain ambient temperature. It never interacts with the Heater or the AC. To ensure that the synthesized coordinator interacts with the heater and AC, we append the following condition to the liveness specification:
 $$ \mathsf{Interact} := \Always\Eventually\mathsf{switchACOn} \wedge \Always\Eventually\mathsf{switchHeatOn}$$
 
 With the modified liveness specification, i.e., $\mathsf{AmbientTemp} \wedge \mathsf{Interact}$, the coordination synthesis procedure returned a 4-processes coordinator depicted in Fig~\ref{Fig:ThermSecond}. It overcomes the deficiency from the previous coordinator by engaging the heater and AC infinitely often.
  This engagement is guaranteed due to the cycle $\mathsf{M_1} \xrightarrow{\mathsf{switchACOn}} \mathsf{M_3} \xrightarrow{\mathsf{switchHeatOn}} \mathsf{M_3}  \xrightarrow{\mathsf{switchHeatOff}} \mathsf{M_2} \xrightarrow{\mathsf{JustRight}}  \mathsf{M_1}  \xrightarrow{\mathsf{switchACOff}}  \mathsf{M_2} \xrightarrow{\mathsf{JustRight}}  \mathsf{M_1} $. 
  In all executions of the coordinated system, the coordinator will be forced to visit state $\mathsf{M_1}$. From here on, the coordinator is forced to take the actions as shown in the cycle. 
 But this is still unsatisfactory, as it allows the Heater and AC to be both switched on at the same time. This can be observed along the execution  $\mathsf{M_0} \xrightarrow{\mathsf{JustRight}} \mathsf{M_1} \xrightarrow{\mathsf{switchACOn}} \mathsf{M_3} \xrightarrow{\mathsf{switchHeatOn}} \mathsf{M_3} \cdots$. Therefore, we add another condition to the liveness specification that enforces that the Heater and AC are not switched on at the same time:
 \begin{align*}
  \mathsf{EnergyEfficient}  := & \neg\Eventually\big((\mathsf{switchACOn}) \wedge (\neg\mathsf{switchACOff} \Untill \mathsf{switchHeatOn})\big) \\
						    \wedge & \neg\Eventually\big((\mathsf{switchHeatOn}) \wedge (\neg\mathsf{switchHeatOff} \Untill \mathsf{switchACOn})\big)
 \end{align*}

 The coordinator obtained by modifying the liveness specification to $\mathsf{AmbientTemp} \wedge \mathsf{Interact} \wedge \mathsf{EnergyEfficient}$ is given in Fig~\ref{Fig:ThermThird}. It is a simple 3-state coordinator. The runtime analysis has been presented in Table~\ref{tab:thermostat}.


\begin{figure}[t]
	\centering
	\begin{minipage}{0.6\textwidth}
		\centering
		\begin{tikzpicture}[shorten >=1pt,node distance=2.5cm,on grid,auto] 
		
		\node[state] (q_0)   { \footnotesize $\mathsf{M_1}$}; 
		\node[state] (q_2) [right=of q_0] {\footnotesize$\mathsf{M_2}$}; 
		\node[state, initial] (q_1) [left=of q_0] {\footnotesize $\mathsf{M_0}$};
		\node[state] (q_3) [below right=of q_0] {\footnotesize$\mathsf{M_3}$};
		\path[->] 
		(q_0)   
		edge [bend left = 10] node {\footnotesize \textsf{switchACOff}} (q_2)
		edge  node [right] {\footnotesize \textsf{switchACOn} } (q_3)	
		edge  node [left] {\footnotesize \textsf{switchHeatOff}} (q_3)

		(q_2)  
		edge [bend left = 10] node {\footnotesize \textsf{JustRight}}   (q_0)
		
		(q_1) edge [bend right = 10] node [sloped,below] {\footnotesize \textsf{Cold}} (q_0)
		edge [bend left = 10] node [sloped,above] {\footnotesize \textsf{JustRight}} (q_0)
		edge  [bend right = 30] node  [left] {\footnotesize \textsf{switchHeatOff}} (q_3)
		(q_3) 	edge [loop right] node { \footnotesize \textsf{switchHeatOn}} ()
		 edge [bend right = 35] node [right] {\footnotesize \textsf{switchHeatOff}} (q_2);
		
		\end{tikzpicture}
		\caption{Coordinator when liveness spec. is $\mathsf{AmbientTemp} \wedge \mathsf{Interact}$}
		\label{Fig:ThermSecond}
	\end{minipage}
	\hfill
	\centering
	\begin{minipage}{0.35\textwidth}
		\centering
			\begin{tikzpicture}[shorten >=1pt,node distance=2.5cm,on grid,auto] 
			\node[state,initial] (q_0)   { \footnotesize $\mathsf{M_0}$}; 
			\node[state] (q_1) [below left=of q_0] {\footnotesize $\mathsf{M_1}$}; 
			\node[state] (q_2) [below right=of q_0] {\footnotesize $\mathsf{M_2}$}; 
			\path[->] 
			(q_0)   
			edge [bend right = 10] node [left] { \footnotesize \textsf{switchACOn}}  (q_1)
			edge node  { \footnotesize \textsf{switchHeatOff}}  (q_2)
			(q_1) 
			edge [bend right = 10 ]node [right] {\footnotesize \textsf{JustRight}} (q_0)
			edge [bend right = 10] node [below] {\footnotesize \textsf{switchACOff}} (q_2)
			
			(q_2)
			edge [bend right = 10] node [above] {\footnotesize \textsf{switchHeatOff}} (q_1)
			edge [loop below] node  { \footnotesize \textsf{switchHeatOn}} ();

			\end{tikzpicture}
			\caption{Coordinator when liveness spec. is $\mathsf{AmbientTemp} \wedge \mathsf{Interact} \wedge \mathsf{EnergyEfficient}$}
			\label{Fig:ThermThird}
	\end{minipage}
\end{figure}

\begin{table}[t]
	\caption{Runtime analysis of thermostat case-study. CSP process is $\mathsf{ENV}$, Safety spec is $\mathsf{False}$.}
	\label{tab:thermostat}
	\begin{tabular}{|c|c|c|c|c|}
		\hline
		\multirow{3}{*}{Liveness spec}                                                                                                                & \multirow{3}{*}{Synthesized Coordinator}                         & \multicolumn{3}{c|}{Run time ( in seconds, timeout=1000s)}                                                                               \\ \cline{3-5} 
		&                                                                 & \multirow{2}{*}{\begin{tabular}[c]{@{}c@{}}Spec. aut. \\ construction\end{tabular}} & \multicolumn{2}{c|}{Sync. synthesis}               \\ \cline{4-5} 
		&                                                                 &                                                                                     & QBF (Default)            & OneHot-SAT              \\ \hline
		$\mathsf{AmbientTemp}$                                                                                                                        & $M = \mathsf{JustRight} \then M$                                & 0.172                                                                               & 1.788                    & 0.429                   \\ \hline
		$\mathsf{AmbientTemp} \wedge \mathsf{Interact}$                                                                                               & Fig~\ref{Fig:ThermSecond}                 & 0.177                                                                               & Timeout                  & 40.152                  \\ \hline
		\multirow{2}{*}{\begin{tabular}[c]{@{}c@{}}$\mathsf{AmbientTemp} \wedge \mathsf{Interact}$ \\ $\wedge \mathsf{EnergyEfficient}$\end{tabular}} & \multirow{2}{*}{Fig~\ref{Fig:ThermThird}} & \multirow{2}{*}{0.188}                                                              & \multirow{2}{*}{Timeout} & \multirow{2}{*}{90.292} \\
		&                                                                 &                                                                                     &                          &                         \\ \hline
	\end{tabular}
\end{table}

\subsubsection{Arbiter: Case study on fairness}
\label{Sec:CaseStudy-Arbiter}

In this  case study we synthesize an arbiter that allocates a shared resource to multiple concurrent processes. The arbiter must guarantee that the resource is not accessed by multiple processes at the same time, and ensure that every requesting process is eventually granted the resource. The study shows how fairness can be incorporated in specifications in order to construct appropriate solutions. 



Formally, a process $\mathsf{P(i)}$ is defined as follows. The process does a cycle of (a). {\em request} to access the common resource, (b) the request is {\em granted} and it accesses the resource, and (c). it {\em releases} the resource.
\[
\mathsf{P(i)} = \mathsf{request.i} \then\mathsf{grant.i} \then\mathsf{release.i} \then \mathsf{P(i)}
\]
For $n>1$, let $\mathsf{ENV_n} = ||_{i\in \{0\dots n-1\}} \mathsf{P(i)}$ denote the environment process, which is the parallel  composition of $n$ copies of $P(i)$. 
Each processes synchronizes with the arbiter (coordinator) on actions $\mathsf{request.i}$, $\mathsf{grant.i}$, and $\mathsf{release.i}$. Therefore $\mathsf{request.i}$, $\mathsf{grant.i}$ , and $\mathsf{release.i}$ actions for all processes are public actions at the interface of $\mathsf{ENV_n}$ with the arbiter.
As these process do not interact internally with each other, and neither do they have an internal action action of their own, $\mathsf{ENV_n}$ does not have any private action and is a fully synchronous environment. 

The safety specification is simply deadlock freedom, indicated by the LTL formula $\mathsf{False}$. The liveness specification includes mutual exclusion and starvation freedom for every process. 
An arbiter can guarantee mutual exclusion if it ensures that each process must release access to the resource before another process is granted access to it:
\[
\mathsf{Mutex}_n := \bigwedge_{i\in [n]} \neg \big(\Eventually (\mathsf{grant.i} \wedge  \big(\neg \mathsf{release.i} \Untill \bigvee_{j \neq i}\mathsf{grant.j} \big)\big)
\]
Starvation freedom is guaranteed if every requesting process is eventually granted access:
\[
\mathsf{StarveFreedom}_n := \bigwedge_{i\in[n]} \Always \big( \mathsf{request.i} \rightarrow \Eventually \mathsf{grant.i} \big)	
\]

At $n=2$, i.e., with two processes, our synthesis procedure generated the following arbiter the liveness specification given by $\mathsf{Mutex}_2 \wedge \mathsf{StarveFreedom}_2$:
\[
M = \mathsf{request.0} \then \mathsf{grant.0} \then \mathsf{release.0} \then M
\]

The arbiter  $M$ trivially satisfies the safety and liveness specifications on the environment process as requests from  processes other than $\mathsf{ P(0)}$ are never enabled. By construction, $M$ is also fair w.r.t. $\mathsf{ENV}_2$. Yet this is clearly not our intended coordinator.  
We observe that  the {\em degree of fairness} for $M$ may differ with respect to individual processes. Specifically, $M$ is fair to process $P(0)$, but is only {\em vacuously} fair to $\mathsf{P(1)}$, as it never offers to synchronize with $\mathsf{P(1)}$.

An even stronger notion of fairness than the one assumed here would rule out such vacuous solutions, but that would require modifications to the algorithm. 
Instead, we modify the liveness specification by appending a condition that forces computations where the coordinator accepts a request from each process:
\[
\mathsf{SimulateStrongFairness}_n := \bigwedge_{i\in[n]}\Always\Eventually \mathsf{request.i}
\]

Our synthesis procedure returned the following arbiter on modifying the liveness specification to $\mathsf{Mutex}_2 \wedge \mathsf{StarveFreedom}_2 \wedge \mathsf{SimulateStrongFairness}_2$

\begin{align*}
	& M_\mathsf{Arbiter2} = 
	\mathsf{request.1} \then M_1 \choice 
	\mathsf{grant.1} \then M_\mathsf{Arbiter2} \choice 
	\mathsf{release.1} \then M_\mathsf{Arbiter2} \\
	& M_1  = 
	\mathsf{request.0} \then M_\mathsf{Arbiter2} \choice 
	\mathsf{grant.0} \then M_1  \choice 
	\mathsf{release.0} \then M_1  	
\end{align*}

This solution satisfies the conditions of desired arbiter for 2-process environment and exhibits a stronger notion of fairness. In fact, it exhibits round-robin behavior despite the specification not including that as a requirement. 
This suggests that using higher-level models such as CSP could have advantages in generating good controllers despite under-specification. 

Similar observations were made for $n=3$. The synthesized coordination program is as follows:
\begin{align*}
  & M_\mathsf{Arbiter3} = 
  \mathsf{grant.1} \then M_1 \choice 
  \mathsf{request.1} \then M_2 \choice 
  \mathsf{release.1} \then M_2 \\
  & M_1  = 
  \mathsf{request.2} \then M_\mathsf{Arbiter3} \choice 
  \mathsf{request.1} \then M_3 \choice 
  \mathsf{grant.1} \then M_3  \\
  & M_2 =
  \mathsf{request.0} \then M_\mathsf{Arbiter3} \choice 	
  \mathsf{grant.0} \then M_2 \choice 
  \mathsf{release.0} \then M_3 \\
  & M_3 =
  \mathsf{grant.2} \then M_1 \choice
  \mathsf{release.2} \then M_2 \choice 
  \mathsf{release.1} \then M_2 \choice 
  \mathsf{request.2} \then M_3
\end{align*}
The runtime analysis has been presented in Table~\ref{tab:arbiter}.
Unfortunately, the implementation did not scale beyond $n=4$. 
But scalability is a challenge in temporal synthesis generally, e.g., recent experiments in~\cite{FFRT17} on the widely-studied synchronous synthesis method show that arbiter synthesis is limited to 4-7 processes.

\begin{table}[t]
	\caption{Runtime analysis of arbiter case-study. For $n>1$, the CSP process is  $\mathsf{ENV_n}$, Safety spec is $\mathsf{False}$, and Liveness spec is $\mathsf{Mutex}_n \wedge \mathsf{StarveFreedom}_n \wedge \mathsf{SimulateStrongFairness}_n$}
	\label{tab:arbiter}
	\begin{tabular}{|c|c|c|c|c|}
		\hline
		\multirow{3}{*}{\begin{tabular}[c]{@{}c@{}}Number of \\ processes ($n$)\end{tabular}} & \multirow{3}{*}{\begin{tabular}[c]{@{}c@{}}Synthesized\\ Coordinator\end{tabular}} & \multicolumn{3}{c|}{Run time (in seconds, timeout=1500s)}                                                                               \\ \cline{3-5} 
		&                                                                                    & \multirow{2}{*}{\begin{tabular}[c]{@{}c@{}}Spec. aut. \\ construction\end{tabular}} & \multicolumn{2}{c|}{Sync. synthesis}               \\ \cline{4-5} 
		&                                                                                    &                                                                                     & QBF (Default)            & OneHot-SAT              \\ \hline
		2                                                                                     & $M_\mathsf{Arbiter2}$                                                                                  & 0.166                                                                              & 8.255                    & 0.6902                 \\ \hline
		3                                                                                     &  $M_\mathsf{Arbiter3}$                                                                & 0.417                                                                              & 1126.266                   & 56.63                  \\ \hline
	\end{tabular}
\end{table}


\subsubsection{Concluding remarks}
\label{Sec:CaseStudy-Verdict}
Our case studies have demonstrated the promise of synthesis in the design of coordination programs. Synthesis raises the level of abstraction and offers considerable flexibility to adjust specifications to alter or introduce requirements. Automated synthesis may even result in simpler and smaller coordination programs than those written by hand. The case studies also show, however, that the prototype implementation is limited to solving small problem instances. We analyze why this is so, which suggests new and interesting directions for further research. 




The automaton-theoretic constructions of Section~\ref{Sec:Symbolic} crucially reduce a coordination synthesis problem to a question of synchronous synthesis. Tables~\ref{tab:illustrativeEg}-~\ref{tab:arbiter} show that the time taken by this procedure is a small fraction of the total time taken by the coordination synthesis procedure. Furthermore, as the size of problem increases, the fraction of time spent inside this reduction phase decreases considerably.

The bulk of the time and space requirements are taken in the synchronous synthesis step. As described earlier in this Section, we use BoSy as the synchronous (bounded) synthesis engine. BoSy encodes a synchronous synthesis problem into a constraint-solving problem, using either SAT or QBF solvers. Although prior investigations suggest that QBF is more effective~\cite{FFRT17}, our domain-specific OneHot-SAT encoding handily outperforms QBF, even though the size of the encoded constraint is larger. This, we believe, is because SAT solvers are more mature than QBF solvers. 





The scalability of the BoSy encoding is influenced by two factors: the number of states in the  co-B\"uchi automaton given as input to BoSy, and the size of the automaton alphabet. In our experiments, the symbolic encoding of the alphabet, coupled with the OneHot-SAT encoding over the inputs keeps the second factor under control. Although the number of states of the co-B\"uchi automaton generated by our transformation procedure is \emph{linear} in the size of the environment and the automaton for the LTL specification, this size appears to be the key limiting factor for the BoSy encoding (and for related tools, such as Acacia+). The limit appears to be automata with more than about $250$ states. 

This limit explains why the arbiter case-study did not scale well beyond 4 processes. The number of states in the flat CSP environment with 4 processes is $3^4$ = 81 , and the automaton for safety are liveness specifications have 1 and 14 states, respectively. As a result, the naive construction of the universal co-B\"uchi automaton consists of $2\cdot 81\cdot 14 = 2286$ states. Even after SPOT's optimization we are left with 703 states, which is beyond the capabilities of BoSy and other synchronous synthesis solvers.

Clearly, better encodings to SAT and QBF, and improvements to the underlying synchronous synthesis tools will help for coordination synthesis. We believe, though, that a gradual improvement in back-end tools can only go so far. It is necessary to bring in higher-level proof concepts such as abstraction and modular reasoning to tackle these large state spaces. For instance,  the arbiter environment is fully symmetric across the processes, which suggests that a combination of symmetry and modular reasoning should help reduce the complexity of synthesis. New research is needed to formulate such algorithms. 


\section{Discussion and Related Work}

The task of coordinating independent processes is one that is of importance in many domains. Coordination must work in the face of concurrency, asynchrony, partial information, noisy data, and unreliable components. Automated synthesis methods could be of considerable help in tackling these challenges.

In this work, we have formalized the synthesis question as follows. The component processes and the coordinator are modeled as CSP processes that interact using the standard CSP handshake mechanism. The specification is given in linear temporal logic, or as a co-B\"uchi automaton, describing desired sequences of interactions. The main contributions are (1) in formulating an expressive model for the problem, as existing models turn out to be inadequate in various ways; (2) in the solution method, which is the first (to the best of our knowledge) to fully handle partial information, and do so for arbitrary LTL specifications; and (3) by giving a complexity-theoretic analysis of the coordination synthesis problem. We show that there are significant practical and complexity-theoretic limitations to scalability, which suggests that research should be directed towards new synthesis methods that can handle large state spaces.


Reactive synthesis has been studied for several decades, starting with Church's formulation of this question in the 1950s (cf.~\cite{thomas2009}). Much of the prior work is on synthesis for the \emph{synchronous, shared-variable} model which, as argued in the introduction, is not a good match for the coordination problems that arise in the motivating domains of multi-robot and IoT coordination. We review the relevant related work in detail below.

\subsubsection*{Reactive Synchronous Synthesis}
Church's formulation of reactive synthesis was inspired by applications to synthesis of hardware circuits, and thus assumes a synchronous model, where the synthesized component alternates between input and output cycles. This question has been thoroughly studied. In particular, the tree-based view of synthesis originates from this line of work, in particular the seminal results of Rabin~\cite{rabin69}.

The seminal work of Pnueli and Rosner~\cite{pnueli1989synchsynthesis} on synchronous synthesis from LTL specifications was followed by the discovery of efficient solutions for the GR(1) subclass~\cite{gr1journal2012,piterman2006synthesis}, ``Safraless'' procedures~\cite{kupferman2005safraless}, and bounded synthesis methods~\cite{compositionalsynthesis,schewe2007bounded}. Those methods, in particular, have been implemented in a number of tools, e.g.,~\cite{bohy2012acacia+,ehlers2010symbolic,ehlers2011unbeast,bosy2017tool,jobstmann-bloem-lily-2006,jtlv2010} and applied in diverse settings (cf.~\cite{d2013synthesizing,kress2010automatic,LOTM13,MS11,MS12}). We re-use the results of this line of research, transforming the asynchronous coordination synthesis problem to an input-output synchronous synthesis problem, which can be solved by several of these tools.

\subsubsection*{Reactive Asynchronous Synthesis}
The seminal work in asynchronous synthesis from linear temporal specifications, also for a shared-variable model, is that of Pnueli and Rosner~\cite{pnueli1989synthesis}. This model is motivated by applications to asynchronous hardware design. In the model, an adversarial scheduler chooses when the synthesized system can sample the input stream. The specification, however, has a full view of all inputs and outputs. The highly adversarial nature of the scheduler coupled with the low-atomicity reads and writes to shared memory makes it difficult, however, to find specifications that are realizable. The original algorithm of Pnueli-Rosner for asynchronous synthesis transforms the problem into a canonical synthesis problem which fits the synchronous model. That algorithm is, however, too complex for practical implementation, although work in~\cite{pnueli2009synthesis,klein2012effective} has devised heuristics that help solve certain cases.

In recent work~\cite{DBLP:conf/cav/BansalNS18} present a significantly simpler and exponentially more compact construction that has been implemented. Our method for incorporating hidden actions is similar to the mechanism used there to compress unobserved input values, while going beyond it to handle other forms of partial knowledge in the richer CSP model. Indeed, one can encode the asynchronous problem as a coordination synthesis problem in CSP.

Schewe and Finkbeiner~\cite{DBLP:conf/lopstr/ScheweF06} model asynchronous synthesis as the design of ``black-box'' processes (unknowns) that interact with known ``white-box'' processes. There are significant differences in modeling and in the solution strategy. Their model assumes that processes are deterministic. In our model, processes may be non-deterministic, which permits the abstraction of complex internal behavior when representing real devices or robots. A second important difference is that of the communication mechanism. In their model, communication is via individual reads and writes to shared variables. As argued in the introduction, the message-passing model of CSP permits a higher degree of atomicity; it is also a better fit for the motivating domains. There are other differences as well: their solution strategy is based on tree automata, as in~\cite{pnueli1989synchsynthesis}, which has proved to be difficult to implement as it requires complex $\omega$-automaton determinization constructions; ours is based on significantly simpler constructions that do not require determinization. 


\subsubsection*{Reactive Synthesis for CSP and similar models}
Manna and Wolper were the first to consider synthesis of CSP processes from LTL specifications~\cite{manna1981synthesis,wolper1982specification}; in their model, environment entities can interact only with the controller and have no hidden actions. This is relaxed in recent work~\cite{DBLP:journals/tse/CiolekBDPU17}, where the environment can be non-deterministic; however, all actions are visible except a distinguished $\tau$ action. The model, therefore, cannot distinguish between multiple hidden actions; neither can the specification refer to such actions. As argued in the introduction, hidden actions arise naturally from information hiding principles; our work removes both these limitations. Web services composition has been formulated in CSP-like models: in~\cite{DBLP:conf/icsoc/BerardiCGLM03}, a coordinator is constructed based on a branching-time specification in Deterministic Propositional Dynamic Logic, but only in a model with full knowledge.

The seminal work of Ramadge and Wonham~\cite{ramadge-wonham-1989} on discrete event control considers a control problem with restricted specifications -- in particular, the only non-safety specification is that of non-blocking. The synthesis algorithm in~\cite{madhusudan-2001} extends the Ramadge-Wonham model with a CSP-like formulation and allows branching-time specifications, but is also based on full knowledge of environment actions. To the best of our knowledge, this work is the first to consider partial knowledge from hidden actions and to allow arbitrary linear-time specifications.

\subsubsection*{Process-algebraic Synthesis}
In a different setting,	work by Larsen, Thomsen and Liu~\cite{DBLP:conf/lics/LarsenT88,DBLP:conf/lics/LarsenX90} considers how to solve process algebraic inequalities of the form $E\parallel X \preceq S$ where $E$ is the environment process, $X$ is the unknown process, $S$ is a process defining the specification, and $\preceq$ is a suitable process pre-order. They formulate an elegant method, using algebraic manipulations to systematically modify $S$ based on $E$ so that the inequality is transformed to a form $X \preceq S'$ which has a clear solution; the transformed specification $S'$ can be viewed as a ``quotient'' process $S/E$. The construction of a quotient can incur exponential blowup, as shown in~\cite{DBLP:conf/concur/BenesDFKL13}. 

\subsubsection*{Other Synthesis Methods}
Synthesis methods for shared-memory systems with branching-time specifications were developed in Emerson and Clarke's seminal work~\cite{emerson1982using}. Wong and Dill~\cite{wong1990synthesizing} also consider the synthesis of a controller under synchronous and asynchronous models for shared-variable communication. In~\cite{DBLP:conf/fossacs/LustigV09}, the authors study a different but related problem of linking together a library of finite-state machines to satisfy a temporal specification.

As discussed in Section~\ref{Sec:Prelims}, synthesis questions may also be solved by producing winning strategies for infinite games. The complexity of games with partial information is studied in~\cite{DBLP:journals/jcss/Reif84}, while symbolic algorithms for solving such games are presented in~\cite{DBLP:journals/lmcs/RaskinCDH07}. These game formulations, however, do not allow for asynchrony. 

There is a large literature on synthesis from input-output examples, representative results include synthesis of expressions to fill ``holes'' in programs~\cite{solar2006combinatorial} and the synthesis of string transformations~\cite{harris2011spreadsheet}. In most such instances, the synthesized program is terminating and non-reactive. One may view a LTL specification as describing an unbounded set of examples. The methods used for example-based synthesis are quite different, however; and it would be fruitful to attempt a synthesis (pun intended) of the example-driven and logic-based approaches.

\begin{acks}                            
  This work was supported, in part, by the \grantsponsor{GS100000001}{National Science
    Foundation}{http://dx.doi.org/10.13039/100000001} under Grant No.~\grantnum{GS100000001}{CCF-1563393}.
  Any opinions, findings, and conclusions or recommendations expressed are those of the author(s) and
  do not necessarily reflect the views of the National Science Foundation.
\end{acks}

\bibliography{refs}

\section*{Appendix}
\renewcommand{\thesection}{A}

\subsection{Simulating The Pnueli-Rosner Asynchronous Model in CSP}
\label{Sec:async-sim}

We show how to encode the Pnueli-Rosner asynchronous model, which is described in the introduction. The input $x$ is encoded by the process definitions below, $X_0$ and $X_1$ represent the two possible values of $x$. Here, $r_0,r_1$ are public actions used to ``read'' the value ($0$ or $1$) of $x$, while $h_0,h_1$ are hidden transitions which model the non-reading points as chosen by an adversarial scheduler. 
\begin{align*}
  & X_0 = r_0 \then X_0 \choice r_0 \then X_1 \choice h_0 \then X_0 \choice h_0 \then X_1 \\
  & X_1 = r_1 \then X_1 \choice r_1 \then X_0 \choice h_1 \then X_1 \choice h_1 \then X_0
\end{align*}
The encoding of output $y$ is simpler. Here, $w_0,w_1$ are public actions that simulate ``writes'' to $y$ of the values $0$ and $1$, respectively.
\begin{align*}
  & Y_0 = w_0 \then Y_0 \choice w_1 \then Y_1  \\
  & Y_1 = w_1 \then Y_1 \choice w_0 \then Y_0 
\end{align*}
Finally, we let the coordinator choose the starting $Y$ state, via the process
\begin{equation*}
  \hat{Y} = w_0 \then Y_0 \choice w_1 \then Y_1
\end{equation*}
But leave the choice of starting $X$ state undetermined, via the process below, where $\tau$ is a private action. 
\begin{equation*}
  \hat{X} = \tau \then X_0 \choice \tau \then X_1 
\end{equation*}
The original specification $\varphi$, is interpreted over an infinite sequence of pairs of $x,y$ values: $(x_0,y_0),(x_1,y_1),\ldots$. This has to be transformed to be interpreted over action sequences that mimic the behavior. Such a ``well-formed'' sequence satisfies several conditions:
\begin{itemize}
\item The sequence begins with a $\tau$ action, then alternates actions from $\Sigma_y =\{w_0,w_1\}$ and $\Sigma_x = \{r_0,r_1,h_0,h_1\}$.
\item Read actions, i.e. $\{r_0,r_1\}$, appear infinitely often.
\item Between successive read actions,  there is at most one change in the write action. (For instance, the sequence $w_0,w_1,w_0$ is not allowed.) This mimics the fact that there is only a single write to $y$ following the latest read.
\end{itemize}
It is easy to program a finite automaton to check these conditions. The automaton for the negated transformed specification $\lNot \varphi'$, is constructed as the product of an automaton that checks well-formedness and the automaton for the negated original specification $\varphi$, converted to analyze sequences of the alternating form defined in the first condition above. The conversion simply replaces an original transition on a pair $(x=i,y=j)$ by a sequence of transitions: the first on $w_j$, the second on $\{r_i,h_i\}$.

With these transformations, it is straightforward to show that $\varphi$ is asynchronously realizable over sequences of the form $(x_0,y_0),(x_1,y_1),\ldots$ if and only if $\varphi'$ is realizable for the environment $E = \hat{X} \parallel \hat{Y}$.  \qed

\subsection{Simulating The Pnueli-Rosner Synchronous Model in CSP}
\label{Sec:sync-sim}
This is a simpler simulation. The process representing the input $x$ is just the following, where there are no internal actions. Thus, all reads are synchronized. 
\begin{align*}
  & X_0 = r_0 \then X_0 \choice r_0 \then X_1  \\
  & X_1 = r_1 \then X_1 \choice r_1 \then X_0 
\end{align*}
The process representing $y$ is as before, as are the initial process definitions for $\hat{Y}$ and $\hat{X}$. The well-formedness conditions are also the same, except that $\Sigma_x$ is now only $\{r_0,r_1\}$. With these definitions and transformations,
it is straightforward to show that $\varphi$ is synchronously realizable over sequences of the form $(x_0,y_0),(x_1,y_1),\ldots$ if and only if $\varphi'$ is realizable for the environment $E = \hat{X} \parallel \hat{Y}$.  \qed

\subsection{Proof of Theorem~\ref{thm:deterministic}}

\begin{theorem*}
A synthesis instance $(E,\varphi)$ is realizable if, and only if, it has a deterministic solution process that has no internal actions. 
\end{theorem*}
\begin{proof}
  
  Let $M$ be a process that is a solution to the instance $(E,\varphi)$. Let $\Delta=\Public \Union \Private$. 

  In the first stage, we eliminate private transitions in $M$ (this set is necessarily disjoint from $\Delta$) to obtain $M'$ that is also a solution but has no internal non-determinism. The states of $M'$ are those of $M$. The transition relation of $M'$ is defined as follows: for states $s,t$ and public action $a$, a triple $(s,a,t)$ is in the transition relation of $M'$ if, and only if, there is a path with the trace $\beta;a;\beta'$ from $s$ to $t$, where $\beta$ and $\beta'$ are both sequences of internal actions of $M$. Consider any maximal computation $x$ of $E \parallel M'$. This computation can be turned into a maximal computation $y$ of $E \parallel M$ simply by restoring the sequences of internal actions of $M$ used in defining each transition of $M'$ (A subtle point is that it is possible for $x$ to be finite and the corresponding maximal sequence $y$ to be infinite if there is an infinite path (a ``tail'') of private transitions of $M$ originating at the final state of $x$.) The two computations have identical traces when projected on $\Delta$, and if $x$ is infinite and fair, so is $y$. As $y$ satisfies $\varphi$ by the assumption that $M$ is a solution, so does $x$.

  In the second stage, we eliminate external nondeterminism from $M'$ to obtain $M''$ that is also a solution but has no external non-determinism. This could be done by determinizing $M'$ using the standard subset construction, but there is a simpler construction that does not incur the worst-case exponential blowup; it simply restricts the transition relation. $M''$ has the same states and initial state as $M'$; however, its transition relation, $T''$, is such that $T''(s,a) \subseteq T'(s,a)$, and $|T''(s,a)|=1$ iff $|T'(s,a)|\geq 1$. Informally, $T''$ chooses one of the successors of $T'$ on action $a$ from each state. By construction, $M''$ is  externally deterministic and has no internal actions. 

   Unlike with the subset construction, the traces of $M''$ obtained by restriction could be a proper subset of those of $M'$, but it is still a solution. Consider any maximal computation $x$ of $E \parallel M''$.  By construction, $x$ is also a maximal computation of $E\parallel M'$. As this satisfies $\varphi$ by assumption, so does $x$. (A subtle point is that if $x$ ends in a dead-end state of $M''$, this state is also a dead-end state of $M'$, as the restriction only removes duplicate transitions.)
\end{proof}

\subsection{Proof of Theorem~\ref{thm:tree-bisim}}

\begin{theorem*}
For a deterministic CSP process $M$, the process $M' = \proc(\fulltree(M))$ is bisimular to $M$.
\end{theorem*}
\begin{proof}
      Let $M' = \proc(\fulltree(M))$, where $\fulltree(M) = (t,\mu_M)$. Define a relation $B$ between the states of $M'$ and those of $M$ by $(\sigma,s)$ is in $B$ iff $s$ is reachable in $M$ by the unique (by determinism) computation with trace $\sigma$. We show that $B$ is a bisimulation relation. By its definition, $(\epsilon,\Start_M)$ is in $B$. Consider a pair $(\sigma,s)$ in $B$.
      \begin{itemize}
      \item Let $(\sigma,a,\delta)$ be a transition in $M'$. By definition, $\delta=\sigma;a$ and  $a \in \mu_M(\sigma)$. By the definition of $\mu_M$, and as $s$ is the unique state reachable by $\sigma$ in $M$, $\mu_M(\sigma)$ is the set of actions enabled at $s$ in $M$. Thus, there is a transition $(s,a,t)$ in $M$, to some $t$. Then $(\delta,t)$ is in $B$. 

      \item Consider a transition $(s,a,t)$ in $M$. As $s$ is reachable by $\sigma$, $t$ is reachable by $\delta=\sigma;a$. By definition, $\mu_M(\sigma)$ is the set of actions enabled at $s$, so it contains $a$. Thus, the transition $(\sigma,a,\delta)$ is present in $M'$, and $(\delta,t)$ is in $B$. 
      \end{itemize}
%
\end{proof}

\subsection{Formulation of $\geneprivate$}
\label{Sec:GenEPrivate}

Predicate $\geneprivate((q,r,e), g, L, (q',r',e'))$  holds if there is a path of 0-length or more consisting only of private transitions from state $e$ in $E$ to $e'$, a run of the $\mathcal{A}_S$ automaton on this path from state $r$ to $r'$, a run of the $\mathcal{A}_L$ automaton on the same path from state $q$ to $q'$,
set $L$ does not intersect with the public actions enabled on all states of the run in $E$,
and $g$ is true if one of the states on the $\mathcal{A}_L$ automaton run is a green (i.e., B\"uchi accepting) state. 
Let $\mathsf{public}(e)$ denote the set of public actions enabled in state $e$ in $E$.
Being a reachability property, this can be defined as the least fixed point of $ZG((q,r,e),g,(q',r',e'))$, where
\begin{itemize}
	\item (Base case) If $q=q',r=r',e=e'$, $L \cap  \mathsf{public}(e) = \emptyset$, then $ZG((q,r,e),g,L, (q',r',e'))$ holds, and $g$ is true iff $q$ is green, and
	\item (Induction) If $ZG((q,r,e),g_0,L, (q_0,r_0,e_0))$ ,  $J((q_0,r_0,e_0),b,(q',r',e'))$ for a private action $b$, and $L \cap \mathsf{public}(e') = \emptyset$, then   $ZG((q,r,e),g,L,(q',r',e'))$ holds, with $g$ being true if $g_0$ is true or $q'$ is green.
\end{itemize}

\subsection{Arbiter for 3-processes}
\label{Sec:SolutionArbiter3}
The coordination program obtained from synthesis of the arbiter for 3-processes is $M_0$


\begin{align*}
  & M_0 = 
  \mathsf{grant.1} \then M_1 \choice 
  \mathsf{request.1} \then M_2 \choice 
  \mathsf{release.1} \then M_2 \\
  & M_1  = 
  \mathsf{request.2} \then M_0 \choice 
  \mathsf{request.1} \then M_3 \choice 
  \mathsf{grant.1} \then M_3  \\
  & M_2 =
  \mathsf{request.0} \then M_0 \choice 	
  \mathsf{grant.0} \then M_2 \choice 
  \mathsf{release.0} \then M_3 \\
  & M_3 =
  \mathsf{grant.2} \then M_1 \choice
  \mathsf{release.2} \then M_2 \choice 
  \mathsf{release.1} \then M_2 \choice 
  \mathsf{request.2} \then M_3
\end{align*}

\end{document}